\documentclass[onefignum,onetabnum]{siamonline220329}
\usepackage[utf8]{inputenc}
\usepackage[english]{babel}
\usepackage{amsfonts}
\usepackage{amssymb}
\usepackage{mathrsfs}
\usepackage[eulergreek]{sansmath}
\usepackage{amsmath}
\usepackage{physics}
\usepackage{subcaption}
\usepackage{chemformula}
\usepackage{chemfig}
\usepackage{graphicx}
\usepackage{acro}
\DeclareAcronym{VAE}{
  short = VAE,
  long = variational auto-encoder}
\DeclareAcronym{TEM}{
  short = TEM,
  long = transmission electron microscope}
\DeclareAcronym{Cryo-SPA}{
  short = Cryo-SPA,
  long = single-particle cryo–electron microscopy}
\DeclareAcronym{Cryo-EM}{
  short = Cryo-EM,
  long = cryogenic electron microscopy}
\DeclareAcronym{LDDMM}{
  short = LDDMM,
  long = large deformation diffeomorphic metric mapping}
\DeclareAcronym{BFGS}{
  short = BFGS,
  long = Broyden–Fletcher–Goldfarb–Shanno}
\DeclareAcronym{ML-EM}{
  short = ML-EM,
  long = expectation-maximization maximum-likelihood}
\DeclareAcronym{MAP}{
  short = MAP,
  long = maximum a posteriori}
\DeclareAcronym{CTF}{
  short = CTF,
  long = contrast transfer function}  
\usepackage{hyperref}
\hypersetup{%
    unicode=true,     
    colorlinks=true, 
    linkcolor=blue,   
    citecolor=red, 
    filecolor=blue,   
    urlcolor=blue
}
\usepackage[nameinlink]{cleveref}
\crefalias{AlgoLine}{line}%
\makeatletter
\let\cref@old@stepcounter\stepcounter
\def\stepcounter#1{%
  \cref@old@stepcounter{#1}%
  \cref@constructprefix{#1}{\cref@result}%
  \@ifundefined{cref@#1@alias}%
    {\def\@tempa{#1}}%
    {\def\@tempa{\csname cref@#1@alias\endcsname}}%
  \protected@edef\cref@currentlabel{%
    [\@tempa][\arabic{#1}][\cref@result]%
    \csname p@#1\endcsname\csname the#1\endcsname}}
\makeatother
\usepackage{etoolbox}

\newtheorem{remark}[theorem]{Remark}

\newcommand{\stochastic}[1]{\mathsf{#1}}
\newcommand{\stochasticgreek}[1]{{\sansmath #1}}
\newcommand{\Cdot}{\,\cdot\,}
\newcommand{\dint}{\,\mathrm{d}}
\makeatletter
\newcommand{\subalign}[1]{%
  \vcenter{%
    \Let@ \restore@math@cr \default@tag
    \baselineskip\fontdimen10 \scriptfont\tw@
    \advance\baselineskip\fontdimen12 \scriptfont\tw@
    \lineskip\thr@@\fontdimen8 \scriptfont\thr@@
    \lineskiplimit\lineskip
    \ialign{\hfil$\m@th\scriptstyle##$&$\m@th\scriptstyle{}##$\hfil\crcr
      #1\crcr
    }%
  }%
}
\makeatother
\newcommand{\Real}{\mathbb{R}}

\newcommand{\opStyle}[1]{\operatorname{\mathcal{#1}}}
\newcommand{\argmin}{\operatorname*{arg\,min}}
\newcommand{\inpro}[3][{}]{ \left\langle #2 , #3 \right\rangle_{#1}}
\newcommand{\traceop}{\operatorname{tr}}
\newcommand{\LpSpace}{\mathscr{L}}
\newcommand{\ContSpace}{\mathscr{C}}
\newcommand{\LieGroup}{G}
\newcommand{\LieGroupMetric}{\operatorname{d}_{\LieGroup}}
\newcommand{\gelem}{g}
\newcommand{\gelemother}{h}
\newcommand{\gcurve}{\gamma}
\newcommand{\LieAlgebra}{\mathfrak{g}}
\newcommand{\AdjLieBracket}{\operatorname{ad}}
\newcommand{\AdjRep}{\operatorname{Ad}}
\newcommand{\aelem}{u}
\newcommand{\aelemother}{\tilde{u}}
\newcommand{\acurve}{\nu}
\newcommand{\dacurve}{\eta}
\newcommand{\momentum}{m}

\newcommand{\ShapeSpace}{V}
\newcommand{\template}{w}
\newcommand{\target}{y}
\newcommand{\GroupAction}{\Phi}
\newcommand{\DataSpace}{Y}
\newcommand{\data}{y}
\newcommand{\dataother}{\tilde{y}}
\newcommand{\stdata}{\stochastic{\data}}
\newcommand{\datanoise}{e}
\newcommand{\stdatanoise}{\stochastic{\datanoise}}
\newcommand{\residual}{r}
\newcommand{\TEMoperator}{\opStyle{A}}
\newcommand{\ForwardOp}{\opStyle{F}}
\newcommand{\EnergyFunc}{\opStyle{E}}
\newcommand{\RegFunc}{\opStyle{S}}
\newcommand{\LossFunc}{\opStyle{L}}
\newcommand{\bbpos}{a}
\newcommand{\bbpostemplate}{\bbpos_0}
\newcommand{\relbbpos}{\amod}
\newcommand{\relbbposother}{\amodother}
\newcommand{\numPep}{N}
\newcommand{\numImgs}{m}
\newcommand{\MapSpace}{X}
\newcommand{\particle}{p}

\newcommand{\AtmModSpace}{A}
\newcommand{\amod}{v}
\newcommand{\amodother}{\amod'}

\newcommand{\amolecule}{\amod^*}

\newcommand{\AtmToMapOp}{\opStyle{B}}
\newcommand{\AtmToMapOpPlane}{\opStyle{B}_{2D}}
\newcommand{\amodplane}{p}
\newcommand{\PoseGroup}{G}
\newcommand{\pose}{\phi}
\newcommand{\stpose}{\stochasticgreek{\pose}}
\newcommand{\AtmModDeforOp}{\opStyle{V}}
\newcommand{\amoddefor}{\vartheta}
\newcommand{\AmodDeforSpace}{\mathfrak{A}}

\newcommand{\MatrixGroup}{\mathbb{M}}
\newcommand{\SE}{\operatorname{SE}}
\newcommand{\SO}{\operatorname{SO}}
\newcommand{\SOLieAlgebra}{\mathfrak{so}}

\newcommand{\GLLieAlgebra}{\mathfrak{gl}}
\newcommand{\rotmat}{\rho}
\newcommand{\residue}{r}
\newcommand{\projbbpos}{p}
\newcommand{\gauss}{\varphi}
\newcommand{\stt}{\stochastic{t}}

\newtoggle{showrevision}
\togglefalse{showrevision} % Don't highlight revisions
\iftoggle{showrevision}{%
  \newcommand{\revision}[1]{{\color{magenta} #1}}
}{%
 \newcommand{\revision}[1]{#1}
}

\headers{Protein structure from cryo-EM}{E. Jansson, J. Krook, K. Modin and O. \"Oktem}

\title{Geometric shape matching for recovering protein conformations from single-particle Cryo-EM data \thanks{{\color{blue1} Funding: }{This work was funded by the Wallenberg AI, Autonomous Systems and Software
Program (WASP), the Knut and Alice Wallenberg Foundation grant WAF2019.0201, and the Swedish Research Council grants 2020-03107 and 2022-03453.}}}
\author{Erik Jansson \thanks{\email{erikjans@chalmers.se}, \textsc{Department of mathematical sciences, Chalmers University of Technology and Gothenburg University, Gothenburg, Sweden}}
\and Jonathan Krook \thanks{\email{jkroo@kth.se}, \textsc{Department of mathematics, KTH - Royal Institute of Technology, Stockholm, Sweden}}
\and Klas Modin \thanks{\email{klas.modin@chalmers.se}, \textsc{Department of mathematical sciences, Chalmers University of Technology and Gothenburg University, Gothenburg, Sweden}}
\and Ozan Öktem \thanks{\email{ozan@kth.se}, \textsc{Department of mathematics, KTH - Royal Institute of Technology, Stockholm, Sweden}}}

\begin{document}
\maketitle

\begin{abstract}
    We address recovery of the three-dimensional backbone structure of single polypeptide proteins from \ac{Cryo-SPA} data.
    \Ac{Cryo-SPA} produces noisy tomographic projections of electrostatic potentials of macromolecules. 
    From these projections, we use methods from shape analysis to recover the three-dimensional backbone structure. 
    Thus, we view the reconstruction problem as an indirect matching problem, where a point cloud representation of the protein backbone is deformed to match 2D tomography data. 
    The deformations are obtained via the action of a matrix Lie group.
    By selecting a deformation energy, the optimality conditions are obtained, which lead to computational algorithms for optimal deformations. 
    We showcase our approach on synthetic data, for which we recover the three-dimensional structure of the backbone. 
\end{abstract}

% REQUIRED

% REQUIRED

\acresetall
\section{Introduction}
\label{sec:intro}
Biological macromolecular complexes (biomolecules), like proteins and nucleic acids, form the molecular machinery that sustains both life and disease. 
An essential part of biomedical research is therefore to understand the functionality of biomolecules, 
for example to design an antiviral drug as a protease that blocks the ``spike'' protein which binds the SARS-CoV-2 virus to our cells. 
Similarly, mRNA vaccines consist of nucleic acids in a large lipid nanoparticle that results in the synthesis of antibodies.
\subsection{Structures of biomolecules}
The functionality of a biomolecule, in particular of a protein, is largely governed by the dynamics of its 3D shape (structure).
Much of structural biology is devoted to developing experimental and computational methods for determining the structure and dynamics of biomolecules.

Historically, the structure of proteins has been obtained through X-ray crystallography.
As the name suggests, the approach requires one to crystallize the protein. 
One can then experimentally measure X-ray scattering data from such crystals. 
Computational methods (phase recovery) combined with experimental protocols are then used to recover the 3D electron density function of the protein from the scattering data.
The final step is to fold the primary structure of the protein (which can be determined from the gene sequence that codes for its synthesis) into the recovered 3D electron density function. This step is fairly straightforward if the density function has sufficient 3D resolution and the folding yields an atomic model of the protein (structure).

X-ray crystallography has been widely successful, and it is currently the standard approach in structural biology for protein structure determination.
It is, however, limited to biomolecules that can be crystallized. 
Furthermore, the approach is designed for recovering a single static structure\footnote{There are approaches, like time-resolved crystallography, that attempt to use X-ray crystallography to recover a dynamically evolving structure.}.
In particular, recovering structural dynamics of biomolecules in tissue (in-situ) or aqueous solution (in-vitro) remains an open problem. 
Development of experimental and computational methods for this purpose is therefore an active research area. 
A promising approach is based on 3D electron microscopy, which we consider next.

\subsection{3D electron microscopy}
Structural biology has in the last decade undergone a revolution with a new generation of computational and experimental methods for recovery and analysis of biomolecular structures. 
Most notably is the development of methods based on \ac{TEM} imaging.

The idea to combine principles from tomography with \ac{TEM} imaging to recover 3D structures of biomolecules can be traced back to \cite{Rosier:1968aa}, which estimated the 3D structure of a biomolecule (bacteriophage T4 tail) by taking advantage of its symmetry.
Since then, the field of 3D electron microscopy has developed along several directions, and most notably towards methods for structural studies of biomolecules that lack specific symmetries.

A key difficulty in using tomographic principles for 3D electron microscopy is to obtain the necessary amount of 2D \ac{TEM} images of the biomolecule from sufficiently many different views.
This is possible for vitrified specimens that consist of many isolated (structurally identical) copies of the biomolecule in varying orientations.
\Ac{Cryo-SPA} is a 3D electron microscopy technique that is specifically designed for such data. 
It was introduced in \cite{Frank:1975aa,Heel:1981aa,Vainstein:1986aa,Heel:1987aa} and has undergone much progress regarding computational methods \cite{Bai:2015aa,Nogales:2015aa,Cheng:2015aa,Carazo:2015aa,Sigworth:2016aa,Cheng:2018aa,Singer:2018aa,Singer:2020aa,Bendory:2020aa,Benji:2020aa}, sample preparation \cite{Liu:2023aa,Venien-Bryan:2023aa,Cheng:2024aa}, and instrumentation \cite{Faruqi:2015aa,Williams:2019aa}.
This development was recognized through the 2017 Nobel Prize in Chemistry, and it has turned \ac{Cryo-SPA} into a major tool in structural biology for studying large rigid biomolecules (homogenous particles) at (near) atomic resolution \cite{Cheng:2015ab,Nakane2020,Herzik:2020aa} that are difficult or impossible to crystallize.
This is essential for both life science research \cite{Nogales:2016aa} and drug discovery \cite{Renaud:2018aa,Robertson:2022aa}.

\subsubsection{Sample preparation and data acquisition}\label{sec:SPA:Data}
The starting point in \ac{Cryo-SPA} is to prepare an aqueous solution that contains many instances of the biomolecule of interest.
Each instance of the biomolecule is referred to as a \emph{particle}.
Cryofixation is then applied to the solution, resulting in a thin, solid slab of vitrified ice.
This method was first introduced in \cite{Dubochet:1982aa,Adrian:1984aa} as a technique to solidify a specimen (needed since specimens are imaged in vacuum) that avoids introducing additional unwanted heavy-metal stains or other contrast enhancements.

The vitrified slab is then imaged in a \ac{TEM} where high-energy electrons scatter against the specimen in vacuum, resulting in several large 2D phase contrast \ac{TEM} images (micrographs).
Image processing techniques (particle picking) are then used to extract 2D image patches (particle images) from such a \ac{TEM} micrograph. 
The particle images are chosen so that each of them represents a (noisy) 2D ``projection image'' of a \emph{single} particle.
Mathematically, one can view the particle images as elements $\data_1, \ldots, \data_\numImgs \in \DataSpace$ where $\DataSpace$ (data space) is the vector space of (digitized) $\Real$-valued $\LpSpace^2$-functions on $\Real^2$ (detector plane).

\subsubsection{3D reconstruction}\label{sec:3Dreco}
The task in \emph{3D reconstruction} is to recover a volumetric representation of the particles from their corresponding 2D particle images.
These 2D images are in \ac{TEM} imaging generated by probing the particle in the specimen with high-energy electrons that scatter against the electrostatic potential generated by the particle and its surrounding aqueous medium.
This potential is mathematically formalized by the \emph{3D map}, which is  a real-valued function on $\Real^3$ that represents this potential (another related function is the electron density map).
3D reconstruction aims to recover the particle specific 3D maps from their corresponding 2D particle images. 
For each particle, there is, however, only a single 2D particle image.
Reconstruction is therefore not feasible unless one makes further assumptions, which in \ac{Cryo-SPA} is to assume that \emph{all} particles are copies of the \emph{same (isolated) biomolecule}.

If the biomolecule is rigid, then the above assumption translates to assuming that all particles have (approximately) identical 3D structure (\emph{homogeneous particles}). 
One can then use tomographic reconstruction techniques to computationally recover the 3D map of the biomolecule (\emph{3D map estimation}), but this requires one to address the challenge that each particle has an unknown 3D orientation/translation w.r.t.\@ the \ac{TEM} optical axis (\emph{pose}).
Additionally, the 2D particle images have \emph{extremely low signal-to-noise ratio}, e.g., the power of the noise can easily be ten times higher than the signal\footnote{A reason for this is that vitrified biological specimens act as weak phase contrast objects when imaged in a \ac{TEM}, so resulting images have very low contrast. 
Furthermore, such specimens are sensitive to radiation damage, so \ac{TEM} images need to be acquired using a low dose (200--400~$\text{electrons}/\text{nm}^2$).}.

\revision{
A further challenge arises when the biomolecule is flexible where particles have varying 3D structure (\emph{heterogeneous particles}) that represent different conformations of the biomolecule. 
Alongside handling unknown poses, one now needs to recover a dynamical 3D map, which is significantly more difficult. \Cref{sec:RelatedWork} has a brief survey on reconstruction methods in this setting.
In this work, however, we restrict our experiments for joint reconstruction and model building to the homogenous setting where the biomolecule has a single conformation. \Cref{sec:outlook} outlines how to extend our approach to the heterogeneous particle setting.
}

\subsubsection{Model building}\label{sec:ModelBuilding}
The final step is \emph{model building}, which refers to building (pseudo) atomic models for the particles.
This is an essential part of using \ac{Cryo-SPA} in biological research, since interpretation and identification of biological functionality often relies on having access to (pseudo) atomic models.

If the biomolecule is a single-chain protein, then model building translates into the task of ``folding'' the primary structure of the protein (which is known) into the particle specific 3D map.
This is in principle possible if the 3D map has sufficient resolution (about 0.4~nm or better), but the process is semi-automatic, and it requires a large degree of human intervention.
Furthermore, estimating the 3D map with such high resolution from \ac{Cryo-SPA} data is only reachable for biomolecules that are rigid (no dynamics) \cite{Nakane2020,Bouvier2022}, so \emph{model building remains an open challenge for flexible biomolecules} (heterogenous particles).

\subsection{Shape matching applied to \acs{Cryo-EM}}\label{sec:CryoEMShapeMatching}
The idea proposed in this paper is to model flexible proteins with geometric shape analysis, and in particular \emph{shape matching}.

In essence, shape analysis concerns finding the energy-minimizing way to \emph{deform} an initial shape (template) to a target shape. 
The framework was developed by Grenander and others in the context of \emph{computational anatomy} \cite{Gr1993,grenander2007}.
It has had various applications, like in medical image analysis  to characterize, find, or understand disease via abnormal anatomical deformations of organs such as the brain or the lungs \cite{Ceritoglu2013,Risser2013}.
Mathematically, new shapes are obtained by deforming a template.
The deformation is often modelled by the action of a group of diffeomorphisms, but in its abstract setting, e.g., as presented by~\cite{Bruveris2013}, the theory allows the deformations to be generated by other means than diffeomorphisms.
More specifically, general shape changes are obtained by letting a Lie group act on a ``shape space'', which can be, e.g., points, curves, surfaces, volumetric maps, or combinations thereof.
\emph{Shape matching} then consists of deforming a template shape so that it matches a target shape, by minimizing a user-defined matching energy.
For more information on shape matching and its underlying mathematical theory, see the book by Younes~\cite{Younes2010} and references therein. 

\subsubsection{Joint 3D reconstruction and model building}\label{sec:Joint3DRecoModel} 
In the following, we show how shape theory allows us to jointly perform 3D reconstruction (\cref{sec:3Dreco}) and model building (\cref{sec:ModelBuilding}). 
These two steps are in \ac{Cryo-SPA} commonly performed sequentially, i.e., one first performs 3D reconstruction (\cref{sec:3Dreco}) and the resulting  reconstructed 3D map(s) is used as input to model building (\cref{sec:ModelBuilding}). 
Both steps involve solving ill-posed inverse problems, so usage of appropriate priors is essential for obtaining reliable and sufficiently accurate solutions.
The key drawback with the sequential approach is that it becomes difficult to make the best use of the a priori information. 
More precisely, much of the a priori information about biomolecular structures and their dynamics is in the form of constraints that apply to (pseudo) atomic models.
These priors apply to model building, and the performance is expected to improve if one can also use (some of) them in 3D reconstruction.
This is possible if one adopts a joint approach in which 3D reconstruction is performed jointly with model building. 
As we show next, with appropriate adaptation of shape matching it is possible to formulate and solve such a joint inverse problem.

Joint 3D reconstruction and model building can now be formalized, as the task  is to estimate particle specific (pseudo) atomic models $\amolecule_1, \ldots, \amolecule_{\numImgs} \in \AtmModSpace$ that represent different conformations of the bio\-molecule from \ac{Cryo-SPA} data $\stdata_1, \ldots, \stdata_{\numImgs} \in \DataSpace$.
Our approach views these conformations as deformations of a \emph{common} (pseudo) atomic model $\amolecule \in \AtmModSpace$ (template).

To proceed, we introduce a series of operators that model various aspects of \ac{TEM} imaging of flexible biomolecules.
First is to model biophysically possible deformations of (pseudo) atomic models with an operator $\AtmModDeforOp \colon \AmodDeforSpace \times \AtmModSpace \to \AtmModSpace$ with $\AmodDeforSpace$ denoting the set that parametrizes such deformations.
Next, particle-specific poses are represented as actions $\PoseGroup \times \AtmModSpace \ni (\pose,\amod) \to \pose.\amod \in \AtmModSpace$ of $\PoseGroup$ on $\AtmModSpace$.
Furthermore, the operator $\AtmToMapOp \colon \AtmModSpace \to \MapSpace$  models the 3D map generated by a (pseudo) atomic model, see \cite[Appendix~B]{Rullgrad:2011aa} for how defining $\AtmModSpace$ in the setting of \ac{TEM} imaging.
Finally, $\TEMoperator \colon \MapSpace \to \DataSpace$ is the operator that models the generation of a 2D image from a 3D map (\ac{TEM} image formation) in the absence of noise.
Under certain assumptions (weak phase object approximation), which hold in \ac{Cryo-SPA}, one can approximate $\TEMoperator$ with a parallel beam ray transform along the \ac{TEM} optical axis followed by a 2D convolution in the detector plane (optics CTF), see \cite{FanelliOktem2008,Rullgrad:2011aa} and \cite[section~4]{Oktem:2015aa}.

We now formalize joint 3D reconstruction and model building as estimating the template (pseudo) atomic model $\amolecule \in \AtmModSpace$ along with particle-specific deformation parameters $\amoddefor_1,\ldots, \amoddefor_\numImgs \in \AmodDeforSpace$ from \ac{Cryo-SPA} data $\stdata_1, \ldots, \stdata_{\numImgs} \in \DataSpace$ where
\begin{equation}\label{eq:Joint3DModelHetro}
\stdata_i := \TEMoperator\Bigr(\AtmToMapOp\bigl(\pose_i.\AtmModDeforOp(\amoddefor_i,\amolecule)\bigr)\Bigl) 
     + \stdatanoise_i
    \quad
    \text{for $i=1,\ldots,\numImgs$ with $\pose_i \in \PoseGroup$ unknown.}
\end{equation}
Then, $\amod_i :=\AtmModDeforOp(\amoddefor_i,\amod)$ for $i=1,\ldots,\numImgs$ are the particle specific (pseudo) atomic models representing the various conformations of the biomolecule.
Furthermore, $\particle_i := \AtmToMapOp(\amod_i)$ will be the corresponding 3D map of the $i$:th particle.

\subsection{Overview of paper and related work}
The paper develops methods from shape analysis for estimating a (pseudo) atomic model of a biomolecule from \ac{Cryo-SPA} particle images (data), i.e., to perform joint 3D map estimation and model building (\cref{sec:Joint3DRecoModel}).
The approach is general and applies to many types of (pseudo) atomic models. 

\subsubsection{The specific setting}
To simplify the presentation, we focus  on biomolecules that are given as single-chain proteins. 
We also assume particle specific poses are known and particle images are from the same conformation (homogeneous particles), so the variant of \cref{eq:Joint3DModelHetro} for joint 3D reconstruction and model building that we consider is 
\begin{equation}\label{eq:Joint3DModelHomo}
\stdata_i := 
(\TEMoperator\circ\AtmToMapOp)(\pose_i.\amolecule) 
     + \stdatanoise_i
    \quad
    \text{for $i=1,\ldots,\numImgs$ with $\pose_i \in \PoseGroup$ known.}
\end{equation}
Finally, we disregard side-chain conformations, so elements in $\AtmModSpace$ represent possible backbone conformations of a protein that is given in terms of its primary structure.

\subsubsection{Structure of paper}
As stated above, we apply shape matching to recover the backbone conformation of a protein from \ac{Cryo-SPA} data of that protein.

The protein is assumed to be in a single conformation, so \ac{Cryo-SPA} data is from homogenous particles.
Next, the shape space represents possible backbone conformations of the linear sequence of amino acids (primary structure) that makes up the protein.
We identify a suitable Lie group that acts on this shape space by folding the backbone in a manner consistent with its possible conformations.
Indeed, instead of acting on protein backbones with diffeomorphisms, we ensure proper, rigid transformations of the chains by restricting shape matching to a finite-dimensional matrix Lie group.
Thereby, we avoid having to work with the infinite-dimensional Lie group of diffeomorphisms, which simplifies not only the theory but also comes with computational benefits. 

Contrary to standard shape matching, however, the template and target shapes in our setting reside in different spaces, which are related to each other via a forward model, as described in the following.

The paper is structured as follows. 
In \cref{sec:lddmm} we present a brief an overview of shape matching in a general setting and outline some gradient-based computational approaches for matching. 
In \cref{sec:cryoemmath} we describe a mathematical model for the protein backbone that is required for defining the shape space. 
We also describe the forward model and perform the computations necessary for shape matching in the specific setting of model building in \ac{Cryo-SPA}. 
In particular, \cref{sec:adapt} provides explicit expressions for gradients needed in shape matching. 
\Cref{sec:experiment} contains a computational example. More specifically, we apply the framework described in the article to capture closed-to-open dynamics of an adenylate kinase protein. 
We end the paper with some conclusions and an outlook on possible future research directions in \cref{sec:outlook}. 
In particular, we discuss how to extend the shape analysis method to the heterogeneous particle setting.
We also mention how to extend the approach to include pose estimation and conformations of side-chains.

\subsubsection{Related work}\label{sec:RelatedWork}
3D map estimation and model building are in \ac{Cryo-SPA} commonly performed sequentially.
The latter amounts to interpreting an estimated 3D map with atomic models, a task that requires expertise and is labor-intensive.
This is in \ac{Cryo-SPA} commonly performed as a separate step applied to the output from 3D map estimation, and many papers deal with automating this post-processing step (\cref{sec:ModelBuilding}).

A different approach that is less explored is to jointly perform 3D reconstruction and model building (\cref{sec:Joint3DRecoModel}).
Here, instead of recovering particle specific 3D maps from \ac{Cryo-SPA} data, one directly recovers the (pseudo) atomic structures of the particle-specific conformations of the biomolecule.
These conformations are given in terms of their (pseudo) atomic structures that are elements in some manifold $\ShapeSpace$.

A natural strategy for joint 3D reconstruction and model building would be to mimic current methods for joint pose and 3D map estimation, namely to set up an iterative refinement scheme that alternates between updating the 3D map and the model. 
The drawback with such a strategy is the difficulty in accounting for the geometry of $\ShapeSpace$ during reconstruction.

Mathematical theory of shape matching combined with variational regularization offers means for performing reconstruction while accounting for the geometry of $\ShapeSpace$.
The overall idea is to view the particles (representing various conformations of the biomolecule) as different deformations of a common template. 
This is consistent with the assumption in \ac{Cryo-SPA} that all particles represent the same biomolecule.
Joint 3D reconstruction and model building can then be seen as matching the ``shape'' of a template against an indirectly observed target (indirect registration). 
An example of indirect registration in tomographic reconstruction was given in \cite{Oktem:2017aa}, which then was extended along various directions \cite{Chen:2018aa,Chen:2019aa,Gris:2020aa,Hauptmann:2023aa}.
Shape matching used in these papers considers shapes of 2D/3D volumetric maps, whereas this paper considers matching shapes of (pseudo) atomic models.
Differential geometric tools used for this are similar to those developed in \cite{Diepeveen:2024aa} for modelling protein dynamics. 

Another differential geometric approach for joint 3D reconstruction and model building is \cite{EsteveYage2023}.
It is based on estimating the torsion and bond angles of the atomic model in each conformation as a linear combination of the eigenfunctions of the Laplace operator in the manifold $\AtmModSpace$ of conformations. 

\revision{
We conclude with briefly mentioning work on reconstruction when the biomolecule is flexible (heterogeneous particles), see \cite{Sorzano2019,Singer:2020aa, Bendory:2020aa} for surveys on this topic.
Many approaches assume the biomolecules adopts a small, fixed, number of conformational states (discrete heterogeneity).
A classification step can then assign a conformation state to each 2D particle image allowing for 3D reconstruction of each conformation. Popular software packages, like RELION \cite{Scheres:2012ab}, Thunder \cite{Hu:2018aa}, Frealign/cisTEM \cite{Grant:2018aa}, SPHIRE \cite{Moriya:2017aa}, and cryoSPARC \cite{Punjani:2017aa} offer this possibility.
Another class of methods assume that the dynamics of the biomolecule can be modelled as a discrete number of independently moving, rigid bodies (multi-body refinement) \cite{Nakane:2018aa,Wong:2014aa,Zhou:2015aa,Bai:2015aa,Ilca:2015aa}.
More elaborate approaches perform reconstruction while heterogeneity is modelled by molecular dynamics simulations \cite{Vuillemot2023a, Vuillemot2023b, Dingeldein2024}. 
A different class of methods assume that the structural variability of the biomolecule forms a manifold that is then learned from data (2D particle images). This is then combined with reconstructing a 3D map at each point on this manifold   \cite{Dashti:2014aa,Schwander:2014aa,Frank:2016aa}.
Closely related approaches model structural variability by means of normal mode analysis \cite{Jin2014} or techniques from Riemannian geometry \cite{Diepeveen:2024aa}. 
Alternatively, one can also represent conformational states of the biomolecules as high-dimensional objects \cite{Lederman:2020aa}. 
The above cited methods rely on some form of model for the dynamics of the biomolecule.
There are also methods that avoid this, like those that parametrise a latent space of conformations by a deep neural network \cite{Zhong:2021aa} or those that use a linear model of variability \cite{Anden:2018aa}.
The latter has been extended using manifold learning techniques \cite{Moscovich:2020aa} that are better adapted for continuous variability. See also \cite{Chen2021, Chen2023, Rosenbaum2024+} for 
further learning based approaches.}

\section{Shape matching}
\label{sec:lddmm}
In general, matching deformable objects refers to the task of aligning selected features of a template object (e.g., landmarks) with corresponding features of a target object.
As discussed in \autoref{sec:CryoEMShapeMatching} above, our premise is to use shape matching, largely following the setup described by \cite{Bruveris2013}.
\revision{For the basic theory of shape matching, we refer to \cite{Younes2010} and for is fundamental geometric underpinnings to \cite{Marsden1999}.}

The starting point is to consider the set $\ShapeSpace$ of deformable shapes (shape space) that we seek to match against each other. $\ShapeSpace$ is often a vector space, but this is not a necessity and $\ShapeSpace$ is sometimes endowed with a manifold structure. 

Deformations used for matching deformable objects in $\ShapeSpace$ against each other are given by acting on $\ShapeSpace$ with a group $\LieGroup$. 
The group action 
$\GroupAction \colon \LieGroup \times \ShapeSpace \to \ShapeSpace$ represents deformations of elements in $\ShapeSpace$ that are parametrized by elements in $\LieGroup$.
Additionally, in shape matching one also assumes that $\LieGroup$ has a manifold structure where the group operations of multiplication and inversion are smooth maps, i.e., $\LieGroup$ is a Lie group with Lie algebra denoted by $\LieAlgebra$.
\begin{remark}
The typical treatment of shape matching considers diffeomorphisms acting on various deformable objects, like point clouds, curves, surfaces, 3D volumes or currents. 
The Lie group $\LieGroup$ is here an infinite dimensional subgroup of the group of diffeomorphisms, e.g., one commonly works  with Fréchet Lie groups of diffeomorphisms.
This results in a rather intricate mathematical theory. 
In contrast, in our adaptation of shape matching for modelling deformation of protein structures, we only need to consider a finite-dimensional Lie group $\LieGroup$.
\end{remark}
\revision{
Indirect matching is the task to match a template $\template \in \ShapeSpace$ against a target that is indirectly observed through data $\data \in \DataSpace$.
This typically arises in an inverse problem setting, where the target is given by (noisy) indirect observations. 
An important component is the data generation process, which is encoded by an operator (\emph{forward model})
\begin{equation}\label{eq:FwdOp}
  \ForwardOp \colon \ShapeSpace \to \DataSpace 
\end{equation}  
that models how a deformable object in $\ShapeSpace$ generates noise-free data in $\DataSpace$.
The indirect registration is performed by finding a group element $\gelem \in \LieGroup$ that parametrizes a deformation $\GroupAction \colon \LieGroup \times \ShapeSpace \to \ShapeSpace$ that minimizes the \emph{matching energy}:
\begin{equation}\label{eq:InDirectRegEnergyInGroup}
  \gelem \mapsto
\LossFunc_{\DataSpace}\Bigl((\ForwardOp \circ \GroupAction)\bigl(\gelem,\template\bigr),\target\Bigr) 
    +  \lambda \LieGroupMetric^2\bigl(\gelem,e \bigr)   \quad\text{for $\gelem \in \LieGroup$.}
\end{equation}
Here, $e \in \LieGroup$ is the identity element, 
$\LossFunc_{\DataSpace} \colon \DataSpace \times \DataSpace \to \Real$ (\emph{data fidelity} functional) quantifies similarity in $\DataSpace$, $\LieGroupMetric \colon \LieGroup \times \LieGroup \to \Real$ is a distance on $\LieGroup$, and $\lambda > 0$ determines the amount of regularization.

An issue with solving \cref{eq:InDirectRegEnergyInGroup} is that the evaluation of the $\LieGroupMetric$-distance itself involves solving an optimization problem, so the minimization in \cref{eq:InDirectRegEnergyInGroup} is a coupled optimization problem over a Lie group. 
One could consider manifold optimization methods \cite{Taylor94,boumal2023intromanifolds}, but there is a simpler way when the distance $\LieGroupMetric$ corresponds to a right-invariant Riemannian metric.
To see this, note first that the minimization of the functional in \cref{eq:InDirectRegEnergyInGroup} can be replaced with minimizing a curve $\gcurve \colon [0,1] \to \LieGroup$ in the group 
\begin{equation}\label{eq:InDirectRegEnergyInCurveOnGroup}
  \gcurve \mapsto
\LossFunc_{\DataSpace}\Bigl( (\ForwardOp\circ \GroupAction)\bigl(\gcurve(1),\template\bigr),\target\Bigr) 
    +  \lambda\int_0^1 \langle \dot\gcurve(t),\dot\gcurve(t) \rangle_{\gcurve(t)} \dint t  
    \quad\text{such that $\gcurve(0)=e$,}
\end{equation}
with $\langle \Cdot , \Cdot \rangle$ denoting a right-invariant metric on $\LieGroup$. 
Next, if $\gcurve$ in \cref{eq:InDirectRegEnergyInCurveOnGroup} is sufficiently smooth, then one can express it as a solution to the flow equation
\begin{equation}\label{eq:flow2}
  \begin{cases}
  \dot\gcurve(t) = \dd R_{\gcurve(t)}(\acurve(t)) &
   \\[0.5em]
  \gcurve(0) = e &
  \end{cases}
  \quad\text{for some $\acurve \colon [0,1] \to \LieAlgebra$,}
\end{equation}
where $\dd R_\gelem\colon \LieAlgebra \to T_\gelem \LieGroup$ is the derivative of the right translation map $R_\gelem \colon \LieGroup \to \LieGroup$ at the identity, i.e., the map $R_\gelem(\gelemother) = \gelemother\gelem$.
Henceforth, $\gcurve^{\acurve} \colon [0,1] \to \LieGroup$ denotes the curve solving \cref{eq:flow2} for a specific choice of $\acurve \colon [0,1] \to \LieAlgebra$.
Note that when $\LieGroup$ is connected, there exists for any $\gelem\in \LieGroup$ a curve $\gcurve^{\acurve}$ such that $\gcurve^{\acurve}(1)=\gelem$.
This correspondence between curves in $\LieGroup$ and $\LieAlgebra$ allows us to rephrase the minimization of \cref{eq:InDirectRegEnergyInCurveOnGroup} over curves in the Lie group $\LieGroup$ as a minimization over curves $\acurve \colon [0,1] \to \LieAlgebra$ in the Lie algebra $\LieAlgebra = T_e\LieGroup$ (which has a vector space structure):
\begin{equation}\label{eq:InDirectRegEnergyInCurveOnAlg}
  \acurve \mapsto
\LossFunc_{\DataSpace}\Bigl( (\ForwardOp\circ \GroupAction)\bigl(\gcurve^{\acurve}(1),\template\bigr),\target\Bigr) 
+  \lambda \int_0^1 \bigl\langle \dot\gcurve^{\acurve}(t),\dot\gcurve^{\acurve}(t) \bigr\rangle_{\gcurve^{\acurve}(t)} \dint t.
\end{equation}

The metric $\langle \Cdot , \Cdot \rangle_\cdot$ can be constructed by  specifying an inner product $(\Cdot,\Cdot)$ on $\LieAlgebra$. 
The metric at the remaining points of $\LieGroup$ is then given by translating the inner product from the right. 
We here consider inner products on $\LieAlgebra$ of the form
\begin{equation}\label{eq:IOperator}
    (\aelem,\aelem) := \bigl\langle \mathbb{I}\aelem,\aelem 
\bigr\rangle_{\LieAlgebra^*,\LieAlgebra}
    \quad\text{for $\aelem \in \LieAlgebra$}
\end{equation}
where $\mathbb{I} \colon \LieAlgebra \to \LieAlgebra^*$ is a symmetric positive semi-definite operator and $\langle \Cdot,\Cdot \rangle_{\LieAlgebra^*,\LieAlgebra}$ denotes duality pairing between the algebra and its dual.
By right-invariance of the metric, we have
\[ 
\int_0^1 \bigl\langle \dot\gcurve^{\acurve}(t),\dot\gcurve^{\acurve}(t) \bigr\rangle_{\gcurve^{\acurve}(t)} \dint t
 = \int_0^1 \bigl\langle \acurve(t),\acurve(t)\bigr\rangle_{e} \dint t .
\]
Hence, the energy functional we seek to minimize during matching reads as 
\begin{equation}\label{eq:general_energy}
  \EnergyFunc(\acurve) := 
  \LossFunc_{\DataSpace}\Bigl((\ForwardOp\circ \GroupAction)\bigl(\gcurve^{\acurve}(1),\template\bigr),\target\Bigr) 
    + \lambda \RegFunc(\acurve)
\end{equation}
where $\lambda > 0$ and $\RegFunc(\acurve)\in \Real$ is the regularization term given by 
\begin{equation}\label{eq:RegFunc}
\RegFunc(\acurve) 
    := \int_0^1 \bigl\langle \acurve(t),\acurve(t)\bigr\rangle_{e} \dint t.
\end{equation}
Choice of $\lambda>0$ in \cref{eq:general_energy} balances the need to penalize unfeasible deformations with \cref{eq:RegFunc} against the need to match the target as quantified by the $\LossFunc_{\ShapeSpace}$-term (data fidelity) in \cref{eq:general_energy}.
}

\subsection{Approaches for minimization}
The task of minimizing the functional in \cref{eq:general_energy} over curves $\acurve \colon[0,1]\to\LieAlgebra$ can be reduced to a dynamical formulation. 
To see this, note first that the data fidelity depends on only $\gcurve^{\acurve}(1)\in \LieGroup$, i.e., on the \emph{final} time point of the curve $t \mapsto \gcurve^{\acurve}(t)$.
Hence, an optimal curve $t \mapsto \acurve(t)$ must follow the dynamics determined by the action functional consisting only of the regularization term of \cref{eq:general_energy} defined in \cref{eq:RegFunc}, i.e.,
\begin{equation}\label{eq:lagrangian}
    \RegFunc(\acurve) 
    :=
    \int_0^1   \langle \momentum(t),\acurve(t)\rangle_{\LieAlgebra^*,\LieAlgebra} \dint t
\end{equation}
where $\momentum  \colon [0,1] \to \LieAlgebra^*$ (\emph{momentum}) is the dual of $\acurve \colon [0,1] \to \LieAlgebra$. 
Due to the right-invariance of the metric, standard methods in geometric mechanics \cite[Theorem~13.5.3]{Marsden1999} yield that $t \mapsto \momentum(t)$ evolves according to the \emph{Euler--Arnold equation}:
\begin{equation}\label{eq:eulerarnold}
    \dot \momentum(t) = \AdjLieBracket_{\acurve}^*\bigl(\momentum(t)\bigr)
    \quad\text{for $t \in [0,1]$.}
\end{equation}
Here, $\AdjLieBracket_{\acurve}^* \colon \LieAlgebra^* \to \LieAlgebra^*$ is defined implicitly by the requirement that the mapping
\[ 
\LieAlgebra \times \LieAlgebra^* \ni \bigl(\acurve(t),\momentum(t)\bigr) \mapsto \AdjLieBracket_{\acurve}^*\bigl(\momentum(t)\bigr) \in \LieAlgebra^* 
\]
is for any $t \in [0,1]$ the adjoint of the Lie bracket on $\LieAlgebra$ with respect to the inner product. 
Thus, performance for approaches that perform matching by minimizing an energy of the form \cref{eq:general_energy} depends highly on the choice of inner product. 
In addition, the group $\LieGroup$ must be known to completely deduce an explicit expression for the mapping $\AdjLieBracket_{\acurve}^*$.

We next consider two algorithms for computing a minimizer, the shooting method and the path method.

\subsubsection{Shooting method}
The shooting method arises from the Euler--Arnold equations \eqref{eq:eulerarnold}. 
The idea is to start with an initial guess of the value of the momentum, and then to integrate \cref{eq:eulerarnold} to obtain the full curve $m\colon [0,1] \to \LieAlgebra^*$. 
From this curve, one then obtains the curve $\acurve\colon [0,1] \to \LieAlgebra$ that can be used to compute the element in $\LieGroup$ acting on the template. 

Note that the regularization energy is conserved along the dynamics of the Euler--Arnold equations \eqref{eq:eulerarnold}. 
Indeed, it holds that 
\begin{equation*}
    \dv{}{t} \bigl\langle \momentum(t),\acurve(t) \bigr\rangle_{\LieAlgebra^*,\LieAlgebra} 
    = 2\bigl\langle \dot \momentum(t),\acurve(t)\bigr\rangle_{\LieAlgebra^*,\LieAlgebra}  
    = 2\Bigl\langle \AdjLieBracket_{\acurve}^*\bigl(\momentum(t)\bigr),\acurve(t)\Bigr\rangle_{\LieAlgebra^*,\LieAlgebra}\!\! 
%    \\[0.5em] 
    =2 \Bigl\langle \momentum(t),\operatorname{ad}_{\acurve}\bigl(\acurve(t)\bigr)\Bigr\rangle_{\LieAlgebra^*,\LieAlgebra}\!\! = 0.
\end{equation*}
Therefore, one only considers the regularization term at the initial momentum, 
\begin{equation*}\tilde{\RegFunc}\bigl(\momentum(0)\bigr)=\bigl\langle \momentum(0),\acurve(0)\bigr\rangle_{\LieAlgebra^*,\LieAlgebra}
\end{equation*}
and the optimal deformation is obtained by optimizing only over possible \emph{initial} momenta.

The above approach is called the shooting method, since one guesses an initial momentum that is then ``shot away'' to obtain the deformation of the template. 
One then evaluates the quality of the match by computing the matching energy and then ``re-aims'' to improve the matching, i.e., update the momentum by taking a step with a scheme given by an optimization method, like gradient descent, \ac{BFGS}-type method or a conjugate gradient type of method.
The shooting method is summarized in \cref{alg:shoot}.
\par\medskip

\begin{algorithm}
\caption{Shooting method for shape matching}
\label{alg:shoot}
\begin{itemize}
    \item \textbf{Initialize:} Suitable guess for initial momentum $\momentum_0^{[0]}$, optimization routine, numerical integration parameters
    \item \textbf{While} $\momentum_0$ changes OR fixed number of steps \textbf{do}:
    \begin{itemize}
        \item Numerically integrate the Euler--Arnold equation \eqref{eq:eulerarnold} to obtain $\acurve\colon[0,1]\to \LieAlgebra$
        \item Numerically integrate \cref{eq:flow2} to obtain an approximation of $\gcurve^{\acurve}(1)$
        \item Evaluate $\LossFunc_{\ShapeSpace}\Bigl( \GroupAction\bigl(\gcurve^{\acurve}(1),\template\bigr),\target\Bigr) + \lambda\tilde{\RegFunc}\bigl(\momentum_0(0)\bigr)$
        \item Compute $\momentum^{[j+1]}$ by taking a step with the optimization routine
        \item \textbf{If needed}:
        \begin{itemize}
            \item Compute $\nabla_{\acurve_0^{[j]}}\LossFunc_{\ShapeSpace}\Bigl( \GroupAction\bigl(\gcurve^{\acurve}(1),\template\bigr),\target\Bigr)$ using \cref{th:gradient_general}
        \end{itemize}
    \end{itemize}
\end{itemize}
\end{algorithm}

\begin{remark}
    For the protein backbone structure, detailed in \cref{sec:cryoemmath} below, the Lie group and Lie algebra has a direct product structure, $G=\SO(3)^{\numPep}$ and $\mathfrak{g} = \SOLieAlgebra(3)^{\numPep}$.
    Furthermore, the inertia operator $\mathbb{I}$ is diagonal with respect to this structure.
    Thereby, the numerical integration steps in \cref{alg:shoot} are naturally parallelized.
\end{remark}

\subsubsection{Path method}
The path method is an alternative to the shooting method that avoids solving the Euler--Arnold equation in \cref{eq:eulerarnold}.

Given a partition of the time interval, $0=t_0 <t_1 <\ldots <t_n =1$, we optimize the values of the momentum at those points, $\momentum_{t_0},\ldots,\momentum_{t_n}$. 
The method is summarized in \cref{alg:desc}
\par\medskip
\begin{algorithm}
\caption{Path method for shape matching}
\label{alg:desc}
\begin{itemize}
    \item \textbf{Initialize:} Time points $t_0, \dots, t_n$, suitable guesses for momentum path $\momentum_0^{[0]}, \momentum_1^{[0]}, \dots, \momentum_n^{[0]}$, optimization routine, numerical integration parameters
    \item \textbf{While} $\momentum_0, \dots, \momentum_n$ changes \textbf{OR} fixed number of steps \textbf{do}:
    \begin{itemize}
        \item Numerically integrate \cref{eq:flow2} using $\momentum_0, \dots, \momentum_n$ to approximate $\gcurve^{\acurve}(1)$
        \item Evaluate $\LossFunc_{\ShapeSpace}\Bigl( \GroupAction\bigl(\gcurve^{\acurve}(1), \template\bigr), \target \Bigr) + \lambda \RegFunc(\acurve)$
        \item \textbf{For} $k = 0, \dots, m$ \textbf{do}:
        \begin{itemize}
            \item Set $\momentum_{t_k}^{[j+1]}$ by taking a step with the optimization routine
            \item \textbf{If needed}:
            \begin{itemize}
                \item Compute $\nabla_{\acurve_{t_k}^{[j]}}\EnergyFunc$ with \cref{th:gradient_general}
            \end{itemize}
        \end{itemize}
    \end{itemize}
\end{itemize}
\end{algorithm}

\subsubsection{Gradient of the objective}
\Cref{alg:shoot,alg:desc} requires computing the gradient of $\EnergyFunc$ in \cref{eq:general_energy} with respect to the curve $\acurve \colon [0,1] \to \LieAlgebra$. 

In the case when $\LieGroup$ is a matrix Lie group, the gradient is straightforward to compute compared to the case of the diffeomorphism group, since there are no spatial derivatives to account for.  
Next, for matrix Lie groups, the group operation is just given by matrix multiplication.
Furthermore, for fixed $\gelem \in \LieGroup$, the tangent left translation $\dd L_{\gelem}\colon \LieAlgebra \to T_{\gelem}(\LieGroup)$ and tangent right translation $\dd R_{\gelem}\colon \LieAlgebra \to T_{\gelem}(\LieGroup)$ are also given by matrix multiplications:
\[
    \dd L_{\gelem}(\aelem) = \gelem \aelem 
    \quad\text{and}\quad
    \dd R_{\gelem}(\aelem) = \aelem \gelem
\quad\text{for a matrix $\aelem \in \LieAlgebra$.}
\]

\Cref{th:gradient_general} explicitly computes the gradient of the energy $\EnergyFunc \colon \ContSpace^{\infty}\bigl([0,1],\LieAlgebra\bigr) \to \Real$ in \cref{eq:general_energy} for matching $\template \in \ShapeSpace$ against $\target \in \DataSpace$ in this specific setting. 
The statement holds for general Lie groups $\LieGroup$ acting on shape spaces $\ShapeSpace$ and a simplified proof for the case when $\LieGroup$ is a matrix Lie group is given in \cref{app:A}.
\begin{theorem}\label{th:gradient_general}
Let the space  $\ShapeSpace$ of deformable objects (shape space) be a manifold and the space $\DataSpace$ of data (data space) is a Hilbert space.
Next, let $\template \in \ShapeSpace$ and $\target \in \DataSpace$ and assume that elements in $\ShapeSpace$ are deformed by acting with a Lie group $\LieGroup$ on $\ShapeSpace$ through an action $\GroupAction \colon \LieGroup \times \ShapeSpace \to \ShapeSpace$.
Furthermore, define $\EnergyFunc \colon \ContSpace^{\infty}\bigl([0,1],\LieAlgebra\bigr) \to \Real$ following \cref{eq:general_energy}, as  
\begin{equation}\label{eq:EnergyIndirectMatching}
  \EnergyFunc(\acurve) := \lambda \RegFunc(\acurve) +
  \LossFunc_{\DataSpace}\Bigl( (\ForwardOp \circ \GroupAction)\bigl(\gcurve^{\acurve}(1),\template\bigr),\target\Bigr) 
\quad\text{for a curve $\acurve \colon [0,1] \to \LieAlgebra$,}      
\end{equation}
where $\gcurve^{\acurve} \colon [0,1] \to \LieGroup$ is given by solving \cref{eq:flow2}, $\RegFunc \colon \ContSpace^{\infty}\bigl([0,1],\LieAlgebra\bigr) \to \Real$ is as in \cref{eq:lagrangian}, and $\ForwardOp \colon \ShapeSpace \to \DataSpace$ is assumed to be differentiable.
When $\LieGroup$ is a matrix Lie group and $\LossFunc_{\DataSpace} \colon \DataSpace \times \DataSpace \to \Real$ (data fidelity) above is given by the squared $\DataSpace$-norm as 
\[
  \LossFunc_{\DataSpace}(\data,\dataother) := 
    \frac{1}{2}\| \data - \dataother \|_{\DataSpace}^2
  \quad\text{for $\data,\dataother \in \DataSpace$,}
\]
then the gradient of $\EnergyFunc$ at $\acurve \colon [0,1] \to \LieAlgebra$ is given by 
\begin{equation}\label{eq:general_gradient}
  \nabla \EnergyFunc(\acurve)(t)
    =\lambda \acurve(t) + 
    \bigl(\AdjRep_{(\gcurve^{\acurve})^{-1}}\bigr)^*
      \biggl(\Bigl(
        \dd\bigl(\ForwardOp \circ \GroupAction(\Cdot,\template) \circ L_{\gcurve^{\acurve}(1)}\bigr)(e)\Bigr)^*\bigl(\residual_{\template,\target}\bigr)
      \biggr),
\end{equation}
where $\residual_{\template,\target} \in \DataSpace$ is a residual term that is defined as
\begin{equation}\label{eq:Residual}
\residual_{\template,\target} := \Bigl( 
  (\ForwardOp \circ \GroupAction)
  \bigl(\gcurve^{\acurve}(1),\template\bigr)
  - \target 
\Bigr).
\end{equation}
\end{theorem}
\begin{remark}
    To interpret the terms/factors in \cref{eq:general_gradient}, note that the first term is the gradient of the regularization term.
    The second term is the gradient of the data fidelity term and is a composition of several expressions. 
    The residual term $\residual_{\template,\target}$ is the difference between the deformed template transformed with $\ForwardOp$ into data space and the target, see \cref{eq:Residual}. 
    To the residual, we apply several adjoint mappings. 
    These appear in various ways. 
    $\ForwardOp \circ \GroupAction(\Cdot,\template) $ arises from the application of the forward model and the deformation of the template. 
    Further, $L_{\gcurve^{\acurve}(1)}$ and $\AdjRep_{(\gcurve^{\acurve})^{-1}}$ appears when computing the first order variation of $\gcurve^{\acurve}(t)$, which is a crucial step in the proof of \cref{th:gradient_general}. 
\end{remark}

While \cref{th:gradient_general} provides all the ingredients for computing the gradient, an explicit expression is only obtainable when one specifies the manifold $\ShapeSpace$ of deformable objects. One also needs to specify a Lie group $\LieGroup$ whose action $\GroupAction \colon \LieGroup \times \ShapeSpace \to \ShapeSpace$ on $\ShapeSpace$ represents deformations. The metric on $\LieGroup$, which is specified by $\mathbb{I}$, is next used to penalize ``unnatural'' deformations. Finally, we have the set $\DataSpace$ (data space) whose elements represent data, which are noisy indirect observations of deformable objects, and the forward model $\ForwardOp \colon \ShapeSpace \to \DataSpace$ models the data in $\DataSpace$ that is generated by a deformable object in $\ShapeSpace$.

\section{Joint reconstruction and model building in \ac{Cryo-SPA}}
\label{sec:cryoemmath}
As shown in \cref{sec:Joint3DRecoModel}, one can view the problem of recovering the backbone conformation of a protein from \ac{Cryo-SPA} data as indirect matching.
This matching problem is, in turn, a special case of shape matching that is outlined in \cref{sec:lddmm} with specific choices of shape space, group action, and data fidelity functional.
Before specifying these, we start with a brief overview of how protein structure is described.

\subsection{Protein structures}
Proteins are large biomolecules that are made up of polypeptides, each consisting of many (typically 100--300) residues.
In a polypeptide, there is a backbone chain consisting of repeating nitrogen-carbon-carbon units.
Residues are small molecules (about 0.5~nm in `diameter') that consist of a backbone part, formed by nine atoms \ch{H2N–CH-C_{$\alpha$}-O2H}, and a side-chain (R-group), which is an amino acid that is covalently bonded to the central carbon atom (denoted by \ch{C_{$\alpha$}}) in the backbone. 
The backbone part of a residue is the same for all proteins, whereas the side-chains vary. For most (not all) proteins, there are 20 possible amino acids that can serve as side-chains, i.e., there are 20 possible residues.
\Cref{fig:protstruct} gives a schematic representation of the structure, together with the first and last polypeptide unit, known as the \ch{N}-terminus and \ch{C}-terminus, respectively. 
In \cref{fig:protstruct_3D}, we further see the orientations of the side-chain and substituents. 
\begin{figure}[hbt]
    \centering
    \begin{subfigure}[t]{0.9\textwidth}
    \centering
    \includegraphics[trim={9.5cm 8.5cm 8.5cm 8.5cm},clip, width=\textwidth]{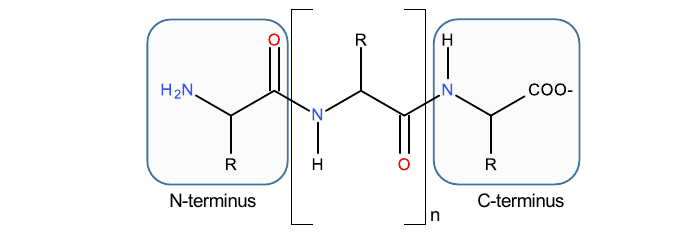}
    \caption{Schematic of a polypeptide backbone consisting of $n$ repeating units of \ch{N-C_{$\alpha$}-C(carbonyl)}.
    In general, the nitrogen has a hydrogen substituent. 
    The first carbon (not shown) is known as the \emph{$\alpha$-carbon}, and it has a hydrogen substituent and a substituent called the side-chain (residue) denoted by the letter $R$.
    The second carbon (carbonyl carbon) is double bonded to an oxygen atom.}
    \label{fig:protstruct}
    \end{subfigure}
    ~
    \begin{subfigure}[t]{0.9\textwidth}
    \centering
    \includegraphics[trim={9.5cm 6.0cm 9.5cm 12.0cm},clip, width=\textwidth]{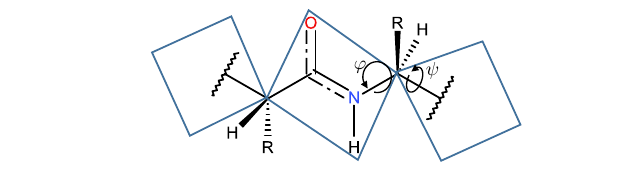}
    \caption{Schematic of the 3D orientation of the peptide units. Each section of the polypeptide between two \ch{C_{$\alpha$}} each are co-planar, whereas the hydrogen atoms (H) and side-chains (R) extend forward or backward out from the plane.}
    \label{fig:protstruct_3D}
    \end{subfigure}
    \label{fig:structure}
    \caption{The ``anatomy'' of the backbone of a protein.}
\end{figure}

There are several ways to describe a protein.
First is its \emph{primary structure}, which is an ordered list of the residues that make up the polymer chain(s). 
Mathematically, this is a labelled linear graph (or string) that has the peptide bonds as edges and the types of residue side-chains (there are 20 possibilities) as labels.
The primary structure can nowadays be easily determined from the gene sequence that codes for the protein.

Next are \emph{secondary structures} that focus on describing the basic folding patterns that the backbone typically displays. 
These patterns are divided into two groups: $\alpha$-helices and $\beta$-sheets.
An $\alpha$-helix refers to a coiling of the backbone that is held together by \ch{H}-bonding between the backbone residues.
The resulting helical structure typically contains about ten amino acids (three turns) but some may have over forty residues.
A $\beta$-sheet is the case when the backbone is arranged along several parallel or antiparallel (with respect to the $\ch{H2N} \to \ldots \to \ch{CO2H}$ direction of the backbone chain) strands that form a pleated (often twisted) sheet-like structure. 
These strands are typically 5–10 amino acids long, and they are connected laterally by \ch{H}-bonds. 
The geometric arrangements for the portions of the backbone that make up these secondary structures can be mathematically characterized \cite{Gromov:2011aa}.

Finally, we have the \emph{tertiary structure} that is the full description of the 3D arrangement of atoms and (covalent) bonds that make up the protein. 
It is determined by the spatial folding of the backbone (backbone conformation) along with the conformations of the peptide-units (side-chains).
This includes specifying the flexible loops that join the fairly rigid ``rod''-like structures formed by the aforementioned $\alpha$-helices and ``plates'' formed by $\beta$-sheets
\footnote{One still lacks a comprehensive formal mathematical language for describing flexible loops \cite{Gromov:2011aa}.}.

As already noted, the backbone is made up in the same way for all natural proteins, but the types and locations of the side-chains differ across proteins.
The identities of these side-chains influence the secondary and tertiary structures \cite[Chapter~3]{Alberts2014}.
In fact, it is commonly accepted by scientists that all information required to specify the tertiary structure of a protein is contained in its primary structure.
To infer/reconstruct the tertiary structure from the primary structure is known as (ab initio) folding.
AlphaFold~2 \cite{Jumper2021} used deep learning to ``address'' this challenge for single-chain proteins. 
This approach was awarded the 2024 Nobel Prize in Chemistry, and it was later extended to multichain proteins in \cite{Evans:2022aa}. 
A further development of this type of structure prediction is AlphaFold~3 \cite{Abramson:2024aa}, which uses deep diffusion-based models to predict structures of not only proteins but also biomolecular complexes that include proteins, nucleic acids, small molecules, ions and modified residues.
However, recovering intrinsic dynamics of isolated (single-chain) proteins is still an open problem.

\subsection{Shape space for protein structures}
Constructing a shape space for a class of biomolecules corresponds to a mathematical formalization of their relevant structural patterns.
Focus here is on the backbone conformation of proteins that consist of a single polypeptide\footnote{Most proteins consist of a single polypeptide, but  some can consist of several polypeptides, like hemoglobin that is composed of four polypeptides.}.

An important requirement on the shape space in our setting is that it should be possible to map an element in the shape space to a 3D map (electrostatic potential) generated by the protein structure encoded by the shape space element.
This indicates that it should be enough to represent the structure as a labelled point cloud, where the labelling is given by the atomic nuclei.
However, when handling protein dynamics, it is highly desirable to exclude deformations that are not biophysical.
This becomes easier if the shape space encodes more of the chemical information, like information about some of the chemical bonds. 
In fact, the dynamics of biomolecules, and proteins in particular, tends to preserve the covalent bonds.
Similarly, local hydrogen bonds that make up secondary structures are also often preserved.  
Hence, the (intrinsic) dynamics of isolated biomolecules is to a large extent governed by the formation and destruction of distant hydrogen bonds.

The primary structure for a protein that consists of $\numPep$ residues can be represented as an array, $(\residue_1, \ldots, \residue_{\numPep})$ where each residue $\residue_i$ can have 20 possible labels (determined by the 20 possible amino acids that a side-chain can be).
Next, the four atoms in the section of the backbone between two \ch{C_{$\alpha$}} atoms are co-planar (see \cref{fig:protstruct_3D}), so the conformation of the backbone can be described with two dihedral angles for each residue.
This means the conformation of the backbone can be described as an array in $\SO(3)^{\numPep}$.
In addition, one needs four dihedral angles to encode the conformation of a side-chain in a residue \cite{Dunbrack:1993aa,Bhuyan:2011aa,Towse:2016aa}, so side-chain conformations can be described as an array in $(\SO(3)^2)^{\numPep}$.
Hence, an all-atom-model applicable to most proteins with primary structure  $(\residue_1, \ldots, \residue_{\numPep})$ is given as an element in $(\SO(3) \times \SO(3)^2)^{\numPep}$.

\subsubsection{Protein backbone conformations}
\label{sec:protein}
The 3D spatial arrangement (structure) of the backbone for a fixed protein is determined by the 3D positions of the \ch{C_{$\alpha$}} atoms in its backbone.
If $\AtmModSpace$ is the set of possible 3D arrangements of the backbone, then its elements are merely arrays of the form $\bbpos=(\bbpos_1,\dots,\bbpos_{\numPep})\in\Real^{3\times \numPep}$ where $\numPep$ is the number of residues and each $\bbpos_i \in \Real^3$ represents the position of the $i$:th \ch{C_{$\alpha$}} atom.
From such a representation, one can compute the electrostatic potential generated by the \ch{C_{$\alpha$}} atoms in the backbone.
With an additional small effort, can also include the contributions to the electrostatic potential from the entire backbone (but disregarding the side-chains).
This is represented by a mapping $\AtmToMapOp \colon \AtmModSpace \to \MapSpace$ and it may be a sufficiently good approximation of the 3D map for the entire protein.
The latter is important when one needs to simulate 2D \ac{TEM} images generated by the protein structure.

Elements in $\AtmModSpace$ do not explicitly encode the covalent bonds in the protein backbone.
These are important since biophysically feasible deformations of the backbone tend to preserve them.
One can implicitly account for these covalent bonds without introducing a computationally demanding graph-type data structure to handle.
The idea is to consider inter-atomic distances between the \ch{C_{$\alpha$}} atoms in the backbone.
These distances are given as
\[
    \delta_{j} := \| \bbpos_{j+1}-\bbpos_j \|\in \Real_+
    \quad\text{for all $j=1,\dots,\numPep-1$,}
\]
and allowable backbone deformations are now those that ensure $\delta_{j}=\delta >0$ for all $j$.
% To simplify expressions, we will henceforth assume $\delta = 1$. 
This assumption that interatomic distances are preserved opens up for representing the backbone structure via the relative positions of the \ch{C_{$\alpha$}} atoms.
To define the shape space, first consider all relative positions of the  \ch{C_{$\alpha$}} atoms in the backbone, which is an array in $\Real^{3\times (\numPep-1)}$:
\[
  (\bbpos_2-\bbpos_1,\dots,\bbpos_{\numPep}-\bbpos_{\numPep-1})
  \in\Real^{3\times (\numPep-1)}
  \quad\text{with}\quad
  (\bbpos_1,\bbpos_2, \dots,\bbpos_{\numPep})
  \in\AtmModSpace.
\]
This could be our shape space, but we also seek a bijection between the shape space and $\AtmModSpace$.
Therefore, we supplement the above relative positions with the absolute position of the first \ch{C_{$\alpha$}} atom, i.e., our shape space $\ShapeSpace$ consists of $\numPep$-arrays of the form 
\[
  \relbbpos=(\bbpos_1,\bbpos_2-\bbpos_1,\dots,\bbpos_{\numPep}-\bbpos_{\numPep-1})
  \in\Real^{3\times \numPep}
  \quad\text{with}\quad
  (\bbpos_1,\bbpos_2, \dots,\bbpos_{\numPep})
  \in\AtmModSpace.
\]
As indicated, the two different representations of the backbone structure in $\AtmModSpace$ and $\ShapeSpace$ are related via the bijective map $\opStyle{M} \colon \ShapeSpace \to \AtmModSpace$ that is defined as
\[
    \opStyle{M}(\relbbpos):=\Bigl(\relbbpos_1, \relbbpos_1+\relbbpos_2, \relbbpos_1+\relbbpos_2+\relbbpos_3,\dots,\sum_{j=1}^\numPep \relbbpos_j \Bigr)
\]
with inverse $\opStyle{M}^{-1} \colon \AtmModSpace \to \ShapeSpace$  given by
\[
    \opStyle{M}^{-1}(\bbpos)=(\bbpos_1, \bbpos_2-\bbpos_1,\dots,\bbpos_{\numPep}-\bbpos_{\numPep-1}).
\]

We next define the Lie group action that deforms the backbone structure of a protein.
Let $\relbbpos, \relbbposother \in \ShapeSpace$ be two different backbone structures of the protein. 
Since these have the same interatomic distances, there exist rotation matrices $\rotmat_1,\dots, \rotmat_{\numPep}\in\SO(3)$ such that $\relbbposother_i=\rotmat_i \relbbpos_i$. 
The deformation of the backbone of the protein can then be represented by acting with the Lie group $\LieGroup := \bigl(\SO(3)\bigr)^\numPep$ on $\ShapeSpace$ via component-wise matrix-vector multiplication:
\[  
\GroupAction \colon \LieGroup \times \ShapeSpace \to \ShapeSpace
\quad\text{where}\quad
\GroupAction(\gelem,\relbbpos) := 
\bigl( \rotmat_1\relbbpos_1, \ldots, \rotmat_{\numPep}\relbbpos_{\numPep} \bigr)
\quad\text{with $\gelem=(\rotmat_1,\ldots,\rotmat_{\numPep} ) \in \LieGroup$.}
\]
This means acting with $\bigl(\SO(3)\bigr)^\numPep$ on the relative atom positions for the \ch{C_{$\alpha$}} atoms in the backbone, which in particular preserves interatomic distances. 
Note also that one cannot independently move a single \ch{C_{$\alpha$}} atom in the backbone, e.g., moving the first one will affect the entire backbone.

\begin{remark}
    An all-atom model for the protein is given by its primary structure, an element in $\AtmModSpace$ (backbone conformation) or $\ShapeSpace$, and conformations of the $\numPep$ residues.
\end{remark}

\begin{remark}
There is a subtle difference between the notion of conformation and structure.
The conformation is the spatial 3D arrangement of the protein up to the pose of the protein, whereas structure refers to the actual position in space (which also includes its pose). 
For biological interpretation, it is enough to work with conformation, but for simulating \ac{TEM} images one also needs the structure since the simulation depends on the pose.
\end{remark}

\subsubsection{The \acs{Cryo-EM} forward operator}
\label{sec:adapt}
The primary structure of the protein is known, and it is used to define the shape space $\ShapeSpace$ of possible backbone structures for the protein. 
These are only indirectly observable by \ac{TEM} imaging.
The \emph{forward operator} models how a deformable object in $\ShapeSpace$ (see \cref{sec:protein}) gives rise to an observable 2D \ac{TEM} image.

2D \ac{TEM} images can be seen as digitized functions in $\DataSpace := \LpSpace^2(\Real^2)$.
We now describe the forward operator used for generating an  element $\data_i \in \DataSpace$ from a backbone structure $\relbbpos_i \in \ShapeSpace$ for $i=1,\ldots, \numImgs$.
The idea is to first map $\relbbpos_i$ to the corresponding 3D point cloud representation in $\AtmModSpace$ via $\opStyle{M} \colon \ShapeSpace\to \AtmModSpace$. 
Next, we create a (approximate) 3D map that serves as input for an operator that models \ac{TEM} image formation.   
One way to do this is to replace each atom in the backbone with a 3D Gaussian, a step that is mathematically formalized by the operator  $\AtmToMapOp \colon \AtmModSpace\to \MapSpace$ where $\MapSpace = \LpSpace^2(\Real^3)$ is the vector space of 3D maps, i.e., functions on $\Real^3$ that can be the electrostatic potential for a protein embedded in vitrified aqueous buffer.
Finally, we generate a 2D \ac{TEM} image from the (approximate) 3D map by an operator $\opStyle{P}\colon \MapSpace \to \DataSpace$ that models \ac{TEM} image formation. 
The simplest such model is given by the parallel beam ray transform that disregards the \ac{TEM} optics and detector response.
The forward operator $\ForwardOp \colon \ShapeSpace \to \DataSpace$ is now given as (see also \cref{fig:forward_model})
\begin{equation}\label{eq:TEMtotFwdOp}
\ForwardOp :=\opStyle{P}\circ\AtmToMapOp\circ\opStyle{M}.
\end{equation}
When $\opStyle{P}$ is the parallel beam ray transform, then one can evaluate \cref{eq:TEMtotFwdOp} in a computationally feasible manner as follows: Perform a geometric projection of the input 3D point cloud along the \ac{TEM} optical axis onto the 2D  \ac{TEM} detector plane. Then, apply the 2D analog of $\AtmToMapOp$ to the 2D point cloud.
\revision{The operator $\opStyle{P}$ in \cref{eq:TEMtotFwdOp} models \ac{TEM} image formation. It is clearly a simplification that omits several aspects of \ac{TEM} imaging, like the optics \ac{CTF} \cite{Oktem:2015aa}. 
If necessary, one can extend $\opStyle{P}$ to also incorporate these.}

\begin{figure}[hbt!]
    \centering
    \includegraphics[trim={3.5cm 1.7cm 2.1cm 3.3cm},clip, width=0.6\textwidth]{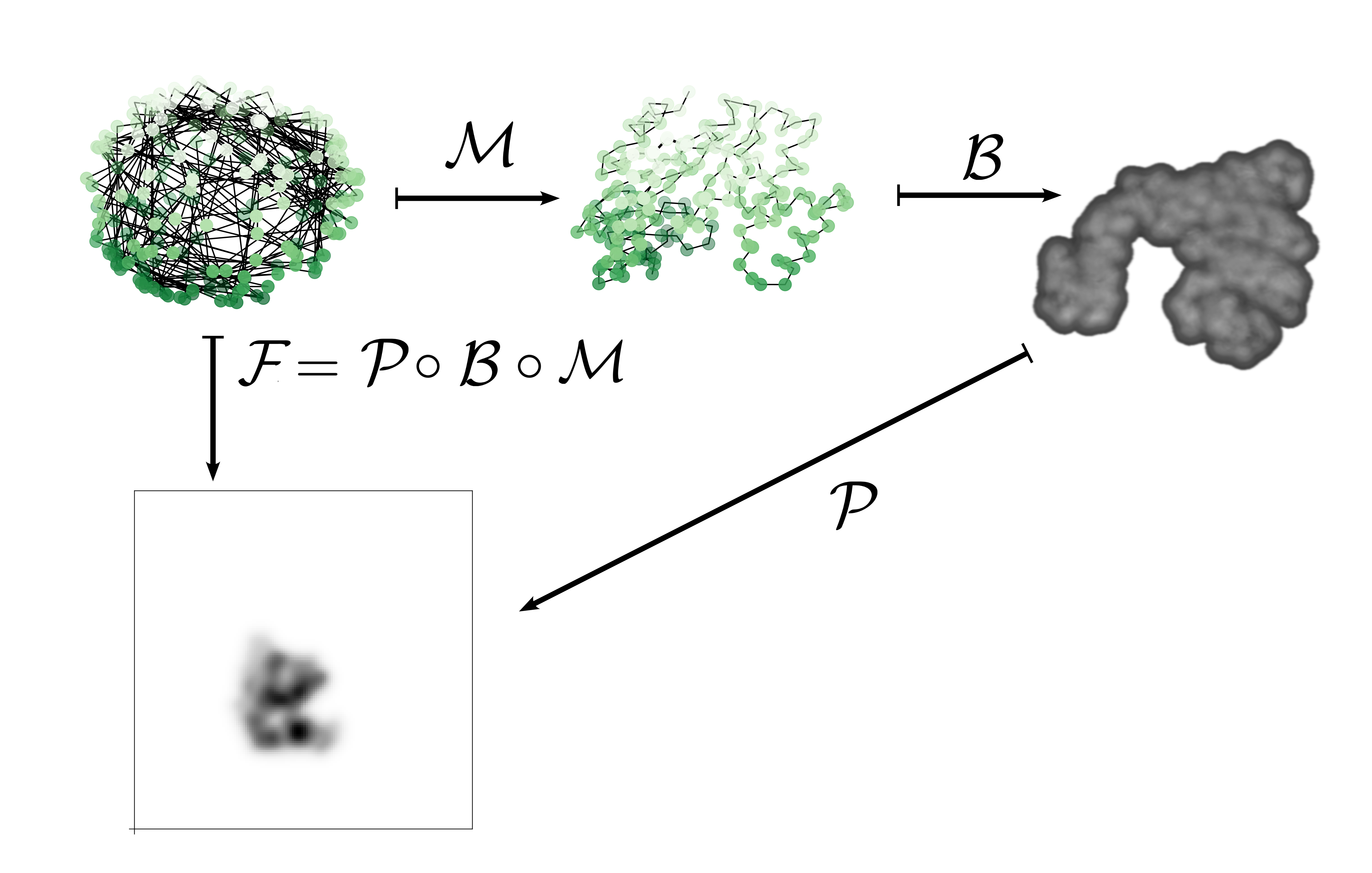}
    \caption{The forward model is built up from three intermediate steps. First is to use $\opStyle{M}$ to map $\ShapeSpace$ (relative positions) to corresponding absolute atomic positions, i.e., elements in $\AtmModSpace$.
    Next, we use $\AtmToMapOp$ to generate a (approximate) 3D map in $\MapSpace$, which is then mapped by $\opStyle{P}$ to a 2D \ac{TEM} image in $\DataSpace$.}
    \label{fig:forward_model}
\end{figure}

\subsection{Shape matching of protein backbones}
\label{sec:protein_specifics}
We here describe how the general theory for shape matching in \cref{sec:lddmm} can be adapted to the setting in \cref{sec:protein,sec:adapt}, i.e., to register a template backbone in $\ShapeSpace := \Real^{3 \times \numPep}$ for a protein by a group action in $\LieGroup := (\SO(3))^\numPep$ against a target backbone that is indirectly observed as a 2D \ac{TEM} image in $\DataSpace := \LpSpace^2(\Real^2)$.
\Cref{sec:protein} offers a motivation for these choices, and solving this indirect registration task will be based on minimizing the energy functional in \cref{eq:EnergyIndirectMatching}. 
Hence, a key topic is to express the gradient in \cref{eq:general_gradient} in this specific setting.

Note first that the Lie algebra $\LieAlgebra$ of the Lie group $\LieGroup$ consists of the direct sum of $\numPep$ copies of $\SOLieAlgebra(3)$, which is the Lie algebra of skew-symmetric matrices. 
The mapping $\GroupAction \colon \LieGroup \times \ShapeSpace \to \ShapeSpace$ that  represents the action of the group $\LieGroup$ on $\ShapeSpace$ is here given by component-wise matrix-vector multiplication.
Further, we set $\mathbb{I} = (\mathbb{I}_1,\ldots \mathbb{I}_{\numPep})$, where each $\mathbb{I}_i$ maps a $\aelem_i \in \SOLieAlgebra(3)$ to its momentum, and let $\mathbb{I}(\aelem) =  \bigl(\mathbb{I}_1(\aelem_1),\ldots \mathbb  I_{\numPep}(\aelem_{\numPep})\bigr)$.
To construct the inner product on the Lie algebra $\LieAlgebra :=\bigl(\SOLieAlgebra(3)\bigr)^{\numPep}$, we first consider the Frobenius inner product on $\SOLieAlgebra(3)$ given by $\inpro[F]{\aelemother}{\aelem} :=-\traceop(\aelemother \aelem)$. 
Any  invertible self adjoint positive map relative to the Frobenius inner product induces an inner product on $\SOLieAlgebra(3)$: $\inpro[]{\aelemother}{\aelem}=\inpro[F]{\mathbb{I}(\aelemother)}{\aelem}$, so we take the inner product on $\LieAlgebra=\bigl(\SOLieAlgebra(3)\bigr)^\numPep$ to be 
\begin{align*}
    \inpro[]{\aelemother}{\aelem} = -\sum_{i=1}^\numPep \traceop\bigl(\aelem_i \mathbb{I}_i(\aelemother_i)\bigr)
\end{align*}

The forward operator $\ForwardOp \colon \ShapeSpace \to \DataSpace$ in \cref{eq:TEMtotFwdOp} can, as described in \cref{sec:protein_specifics}, be computed more efficiently as follows: Use $\opStyle{M}(\relbbpos) = \bbpos$ to map the relative positions of the backbone \ch{C_{$\alpha$}} atoms $\relbbpos \in \ShapeSpace$ to the corresponding 3D cloud of atomic positions $\bbpos \in \AtmModSpace$.
\revision{Then, we geometrically project the atomic positions along the \ac{TEM} optical axis down to the 2D detector plane. 
Mathematically, the $i$:th component of the projection mapping $\pi \colon \AtmModSpace \to \Real^{2 \times \numPep}$ is given by
\[
\bigl( \pi(\bbpos) \bigr)_i 
  =  \begin{pmatrix} 
      \bbpos_i^x \\[0.5em]
      \bbpos_i^y
    \end{pmatrix} \in \Real^{2},
\]
where $\bbpos_i = (\bbpos_i^x,\bbpos_i^y,\bbpos_i^z) \in \Real^3$ denotes the $i$th component of $\bbpos$ and $i = 1,\ldots,\numPep$. 
Note that any projection onto a plane can be viewed as a projection onto the $(x,y)$-plane in the right coordinates. 
Therefore, by transforming coordinates, using the above projection and its adjoint, and transforming back into the original coordinates, all projections as well as their adjoints can be obtained. }
Finally, the projected atomic positions are mapped into $\LpSpace^2(\Real^2)$ by $\AtmToMapOpPlane \colon \Real^{2\times \numPep}\to \DataSpace$, that sends a 2D point cloud $\amodplane\in \Real^{2\times \numPep}$ to a linear combination of 2D Gaussians 
\[\AtmToMapOpPlane(\amodplane)(\Cdot) =
    \sum_{i=1}^{\numPep} s_i (\tau_{\amodplane_i} \circ \gauss_{\sigma_i})(\Cdot),
\]
where $s_1,\dots,s_{\numPep}$ are weights describing the relative importance of each atom and $\sigma_1,\dots,\sigma_{\numPep}$ are parameters determining the width of the Gaussians.
These parameters depend on the atom represented by the point $\amodplane_i$ in the point cloud, with $\gauss_{\sigma_i}$ denoting the function 
\[ \gauss_{\sigma_i}(x,y) := \frac{1}{2\pi\sigma_i^2}\exp\biggl(-\frac{x^2+y^2}{2\sigma_i^2}\biggr)
\]
and $\tau_{\amodplane_i} \colon \DataSpace \to \DataSpace$ denoting translation by the point $\amodplane_i$.
The forward operator in \cref{eq:TEMtotFwdOp} can then be expressed as 
\[
\ForwardOp = \AtmToMapOpPlane \circ \pi \circ \opStyle{M}. 
\]

Consider now the energy functional $\EnergyFunc \colon \ContSpace^{\infty}\bigl([0,1],\LieAlgebra\bigr) \to \Real$ in  \cref{eq:EnergyIndirectMatching}.
\Cref{cor:gradient_so3} in \cref{app:B} computes its gradient by applying \cref{th:gradient_general} to this setting. 
This yields a closed-form expression for the gradient of $\EnergyFunc(\acurve)$ that is useful in performing joint 3D reconstruction and model building (\cref{sec:Joint3DRecoModel}).

\begin{remark}\label{rem:double_proj}
    In \cref{sec:protein_specifics} we have assumed that the particle is in a single \revision{pose, i.e., we have included only one 2D image}. 
    It is possible to include \revision{multiple poses } into this framework with minimal changes, \revision{simply by using more than one image}. 
    Indeed, to have two \revision{images}, one takes two projection functions $\pi_1 \colon \AtmModSpace \to \Real^{2\times \numPep}$ and $\pi_2 \colon \AtmModSpace \to \Real^{2\times \numPep}$
    and proceeds to work with two forward models, 
    \[ \ForwardOp_1 = \AtmToMapOpPlane\circ \pi_1 \circ \opStyle{M}
      \quad\text{and}\quad
      \ForwardOp_2 = \AtmToMapOpPlane\circ \pi_2 \circ \opStyle{M}.
    \]
    One then sets $\tilde{\DataSpace} = \DataSpace \oplus \DataSpace$ and takes the target data to be $(\data_1,\data_2)$ so the data fidelity is
    \begin{align*}
    (\data_1,\data_2) \mapsto
        \frac{1}{2}\biggl(
        \Bigl\|
          \ForwardOp_1\bigl(\GroupAction(\gcurve_{\acurve(1)},\bbpostemplate)\bigr)-\data_1 
        \Bigr\|_{\LpSpace^2}^2 
        + \Bigl\|                    \ForwardOp_2\bigl(\GroupAction(\gcurve_{\acurve(1)},\bbpostemplate)\bigr)-\data_2 
        \Bigr\|_{\LpSpace^2}^2
        \biggr).
    \end{align*}
    The gradient of this modified functional is computed by taking the sum of the gradients given by \cref{cor:gradient_so3} applied to each term. 
    The extension to an arbitrary number of poses is straightforward. 
\end{remark}

\section{Numerical experiment on adenylate kinase}
\label{sec:experiment}
We now apply our method for indirect matching of 3D protein backbone structures. 
The focus of this paper is to introduce shape matching to the setting of protein data. 
Therefore, we perform illustrative proof-of-concept computations, which are deliberately kept simple.\footnote{All results presented in this section are produced using   code available at \url{https://github.com/erik-grennberg-jansson/protein-lddmm}. } 
\begin{figure}[!htb]
     \centering
     \begin{subfigure}[t]{0.4\textwidth}
         \centering
         \includegraphics[trim={7cm 7cm 6.5cm 7cm},clip, width=\textwidth]{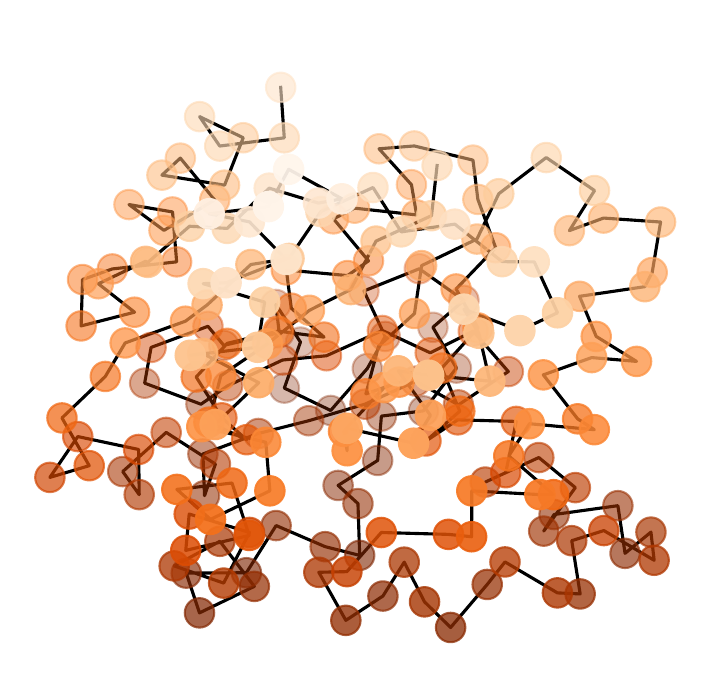}
         \caption{Initial, closed, protein backbone conformation. Obtained from the AlphaFold database entry A0A7H9QZH8.}
         \label{fig:initial_bb}
     \end{subfigure}
     ~
     \begin{subfigure}[t]{0.4\textwidth}
         \centering
         \includegraphics[trim={7cm 8cm 6cm 7cm},clip, width=\textwidth]{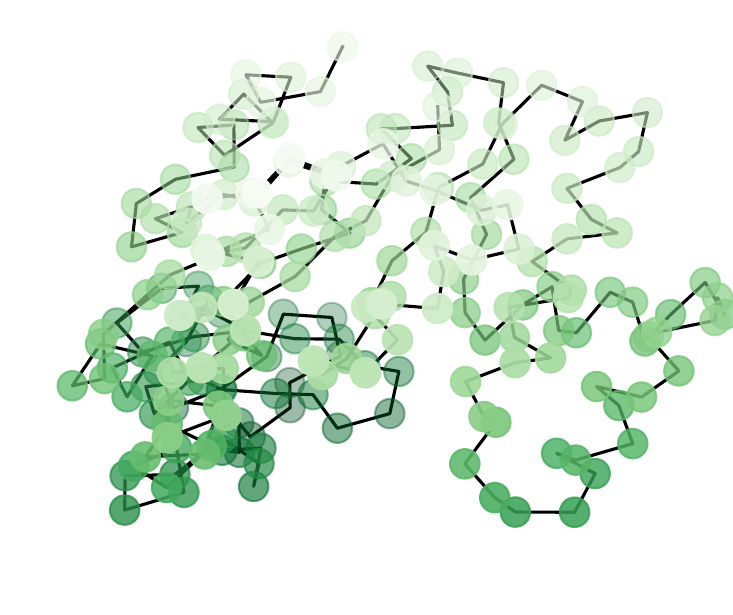}
         \caption{Open, protein backbone conformation to be recovered. Obtained from molecular dynamics simulation \cite{Beckstein2018,Seyler2015}.}
         \label{fig:target_bb}
     \end{subfigure}
     \caption{The initial and final conformation of the \ch{C_{$\alpha$}} atoms in the backbone of the closed-to-open adenylate kinase deformation. The figures depict the elements in $\AtmModSpace$.}
     \label{fig:target_template}
\end{figure}

The protein considered in this numerical experiment is an adenylate kinase protein, which is chosen as it has a clear closed-to-open transition.
Intuitively, this means that the protein initially has a ``closed lid'' that gradually opens, see \cref{fig:target_template} for an illustration of 3D arrangement of the open and closed backbone conformation.

\begin{figure}
    \centering
    \includegraphics[ width=0.8\textwidth]{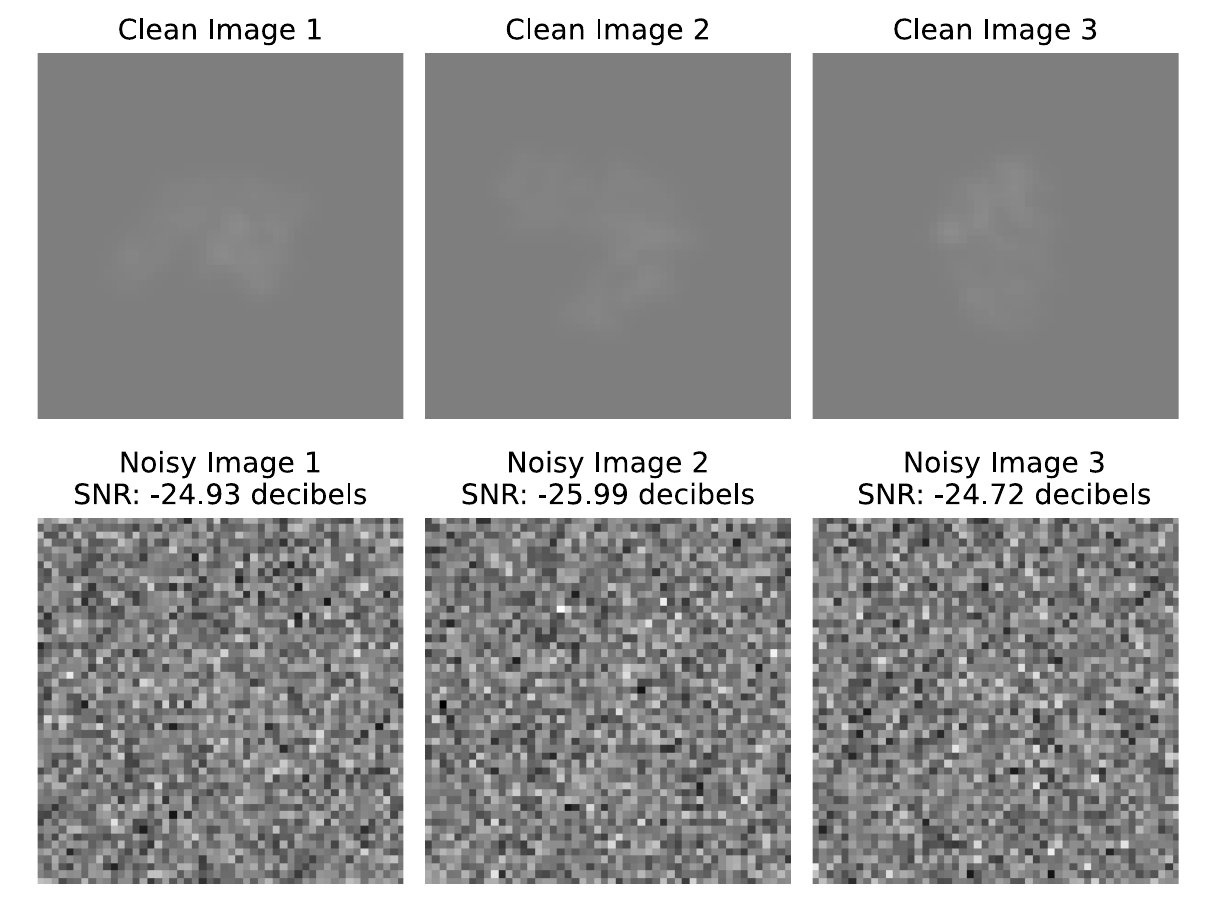}
    \caption{\revision{The data (in  the data space $\DataSpace$) used to deform the template together with intensity value scales to indicate the signal strength relative to the noisy images. The  three noisy images are obtained by  applying the forward model to the protein conformation in  \cref{fig:target_bb} and adding Gaussian noise with standard deviation $1.0$ to each pixel. }}
    \label{fig:data}
\end{figure}

To construct the target, we apply the forward model to the final frame of a protein trajectory from \cite{Beckstein2018,Seyler2015}, which was obtained by molecular dynamics simulations. 

To construct the template, we use  data from the AlphaFold database corresponding to the closed conformation\footnote{\url{https://alphafold.ebi.ac.uk/entry/A0A7H9QZH8}} \cite{Jumper2021,Varadi2021}. 
As AlphaFold data is not necessarily aligned with the molecular dynamics trajectory, we use Procrustes analysis to rotate and mirror the AlphaFold point cloud to align with the  first frame of the protein trajectory. 
Briefly, Procrustes analysis  computes the orthogonal linear transformation aligning one point cloud with another that is optimal in the sense that it minimizes the \emph{Procrustes distance}.  For details, see \cite{Gower1975,Kendall1989}. 
In future work, we aim to estimate the alignment of the AlphaFold data from the observed images. 
This would entail solving yet another inverse problem. 
Techniques to achieve this goal could be achieved by extending the method introduced in \cite{Diepeveen2023}. 

We only consider the \ch{C_{$\alpha$}} atoms in the protein backbone. 
We aim to reconstruct the structure in the final frame of the trajectory based off indirect observations of it, as described in \cref{sec:protein_specifics}. 
It is not a priori clear which, and how many, projections should be selected. 
Note that in the absence of noise, three projections: on the $(x,y)$-plane, the $(x,z)$-plane, and the $(y,z)$-plane capture all information about the protein.  
Therefore, a naive first approach is to proceed as in \cref{rem:double_proj}, but this time with three elements in $\DataSpace$ instead of two.

When we apply $\AtmToMapOpPlane$ to the projected point clouds, we use a width (standard deviation) of $2.0$. 
The images are in practice discretized as pixel images with resolution $50 \times 50$ pixels. 
We then add, to each pixel in the images, Gaussian observational noise  with standard deviation $1.0$.
With this choice of observational noise, the signal-to-noise ratio, computed by dividing the variance of the pixel image without noise with the noise variance, is on average, $0.003$, or expressed in decibels, $-25.181$ dB. \revision{ The signal-to-noise ratio is in line with, or somewhat lower, than the ones typically present for cryo-EM \cite{Bendory:2020aa}. 
The data, together with intensity value scales, is depicted in \cref{fig:data}.
}

To compute the matching, we use the path discretization method of \cref{alg:desc} so that the curve of algebra elements  $\acurve \colon [0,1] \to \bigl(\SOLieAlgebra(3)\bigr)^\numPep$ is discretized  at  fixed time grid points $t_n = n \delta t$ where $\delta t = 0.01$ and $n = 0,1,2, \ldots, 100$. 
Further, we make a non-informed of choice of the momentum mapping and set $\mathbb{I}$ to be the identity.

As described in \cref{sec:protein_specifics}, we must compute the deformation of the template by determining the path of group elements $\rotmat_0,\rotmat_{\delta t},\ldots \rotmat_{1}$ determined by $\acurve_0,\acurve_{\delta t},\ldots \acurve_{1}$. This is done by applying the Lie--Euler integrator with a step size of $\delta t =0.01$ to the flow equation $\dot \rotmat = \rotmat\acurve$,  
where $\rotmat_0$ is initialized as the identity element of $(\SO(3))^\numPep$. 
Applying this integrator guarantees that $\rotmat_1$ is an element of $(\SO(3))^\numPep$  \cite{IsMuKaNoZa2000}. 
\begin{figure}
    \includegraphics[trim={0cm 0cm 0cm 0cm},clip, width = \linewidth]{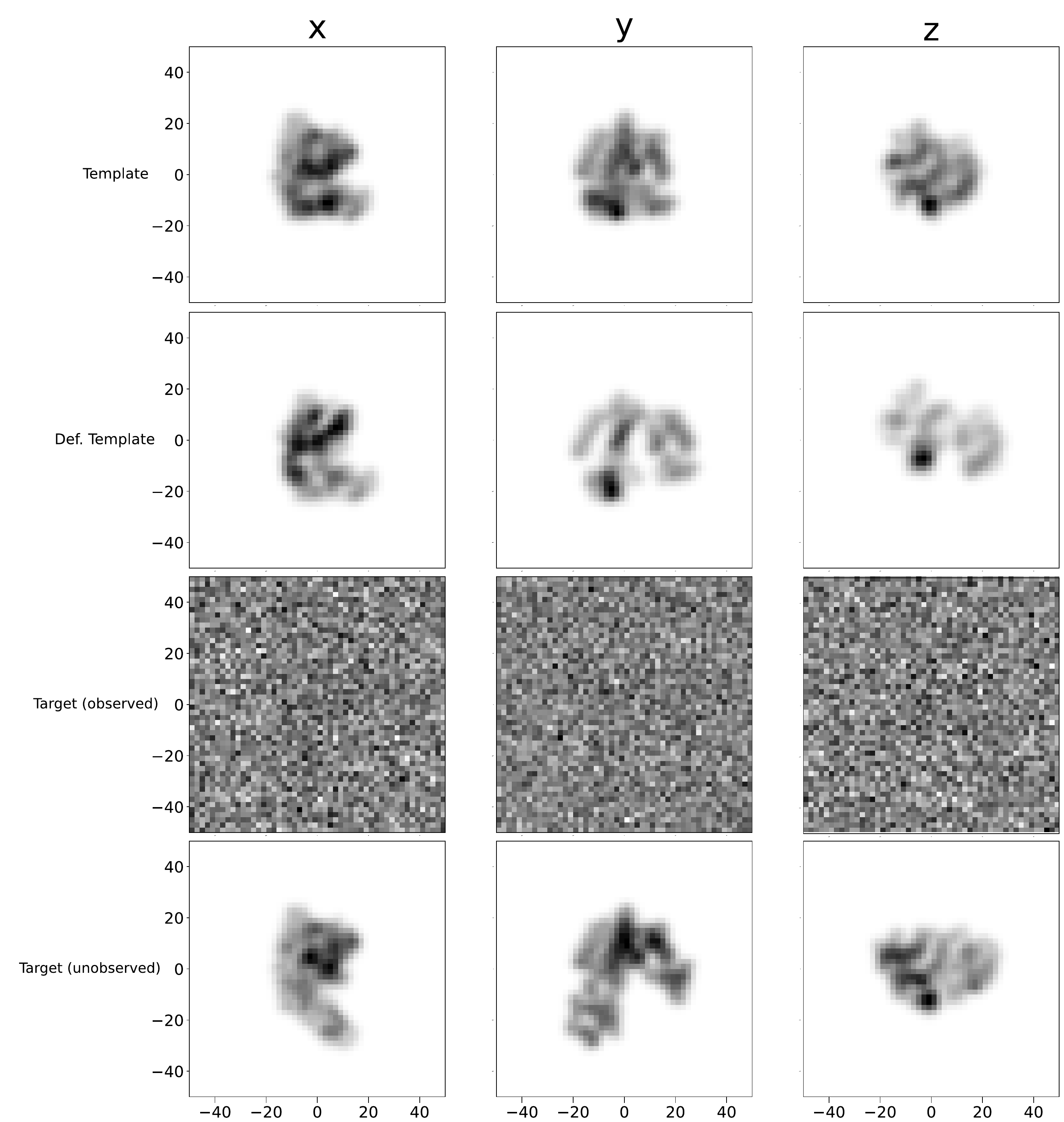}
    \caption{The results of applying the indirect shape matching algorithm to the template (\cref{fig:initial_bb}) using projections along the $x$, $y$ and $z$ axes. Note that while not perfect, the deformation along these axes decently captures the changes between the initial and target conformation. }
    \label{fig:results_projection}
\end{figure}
Finally, we compute the gradient of $\EnergyFunc$ with respect to $\acurve_t$, denoted by  $ \nabla_{\acurve_t} {\EnergyFunc}$ as outlined in \cref{sec:protein_specifics}. 
We then optimize $\EnergyFunc$ by using the L-BFGS-B algorithm implemented in the Python scientific computing package SciPy, \cite{2020SciPy-NMeth}. 
This algorithm is a limited-memory version of the BFGS algorithm, see \cite{Byrd1995} for details. 
To obtain a suitable initial guess for the optimizer, we run $50$ steps of gradient descent, with a step size of $10^{-5}$. \revision{The specific optimization procedure was selected empirically through iterative tuning. The energy landscape is likely highly non-convex with many local minima, and we suspect this contributes to some of the convergence variability observed, especially when only a few projection images are used. Nonetheless, the chosen setup consistently produced strong results in our experiments, also when varying the number of projection images, suggesting it is effective in practice. In future work, it may worthwile pursuing solutions to this problem based on for instance measure-based lifting as in \cite{Diepeveen2023}.} 

\Cref{fig:results_projection} illustrate $\ForwardOp_1$,  $\ForwardOp_2$ and $\ForwardOp_3$ applied to the non-deformed template shown in \cref{fig:initial_bb}, $\ForwardOp_1$,  $\ForwardOp_2$ and $\ForwardOp_3$ applied to the deformed template, the directly observed noisy target images as well as $\ForwardOp_1$,  $\ForwardOp_2$ and $\ForwardOp_3$ applied to the non-directly observed target shown in \cref{fig:target_bb}. 
All of these images are in $\DataSpace$. 
It is clear that we manage to capture the closed-to-open movement of the protein, at least along the directions of the projections. 
\begin{figure}
    \centering
    \includegraphics[ trim = {0 0cm 0 0.0cm}, clip, width=\textwidth]{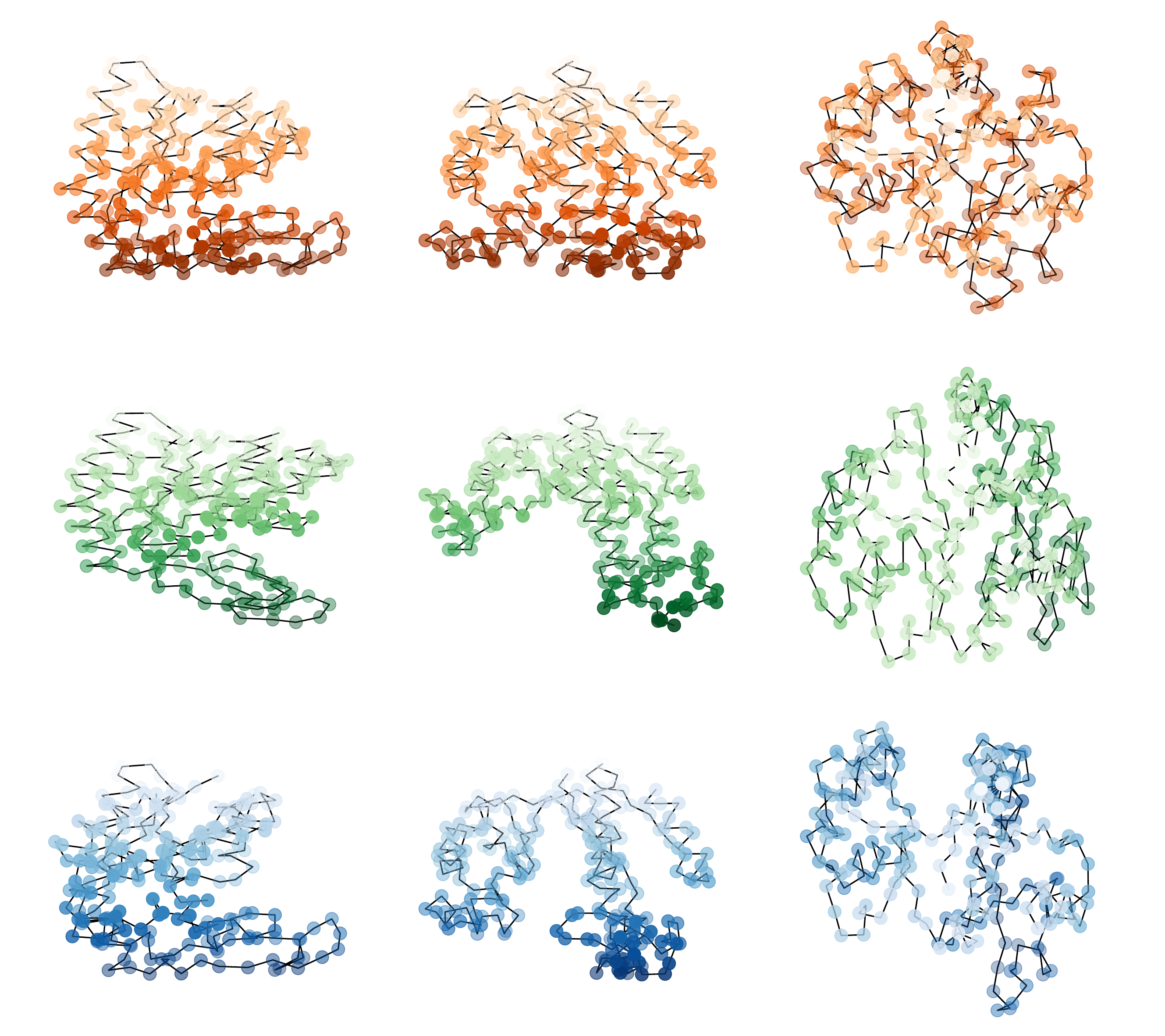}
    \caption{The reconstructed 3D arrangement (blue) together with the template (orange) and non-directly observed target (green), viewed along the same three projection angles as in \cref{fig:results_projection}. Rows depict template, target and deformed template  respectively. Note that we generally capture the deformation, especially the large closed-to-open movement in the $(x,z)$-plane (middle column).}
    \label{fig:xyz_3d_view}
\end{figure}

Consider now the reconstructed 3D protein conformation in \cref{fig:poor_reconstruction}, i.e., the reconstruction depicted in the set $\AtmModSpace$ of possible 3D arrangements of the backbone.

It is obvious that this is a poor reconstruction of the target conformation. 
However, the key takeaway is that we capture the deformation in the directions we project, see \cref{fig:xyz_3d_view}, in which the reconstruction is depicted in the $3D$ arrangement, i.e., as an element in $\AtmModSpace$. 
\begin{figure}[!htb]
     \centering
     \begin{subfigure}[t]{0.4\textwidth}
      \centering
         \includegraphics[trim={7cm 7cm 6.5cm 7cm},clip, width=\textwidth]{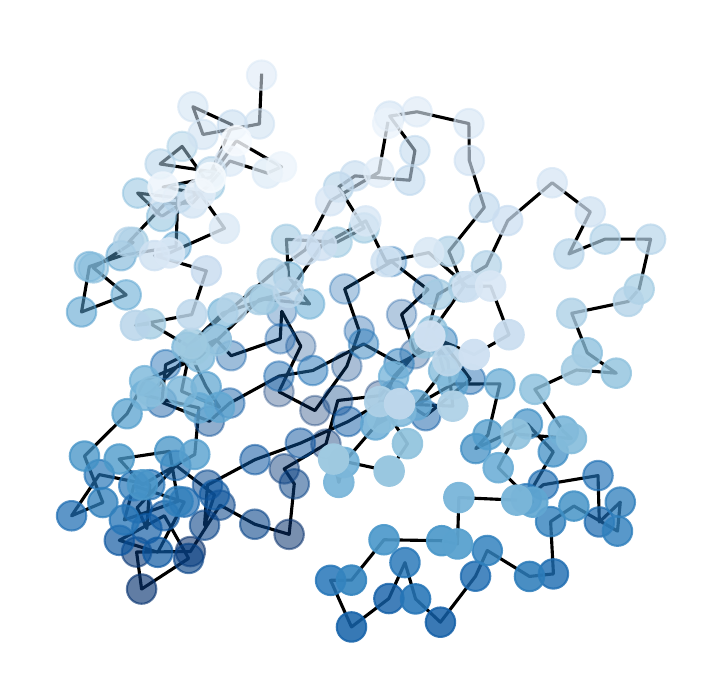}
         \caption{3 projection images.}
         \label{fig:poor_reconstruction}
     \end{subfigure}
     ~
     \begin{subfigure}[t]{0.4\textwidth}
         \centering
         \includegraphics[trim={7cm 8cm 6cm 7cm},clip, width=\textwidth]{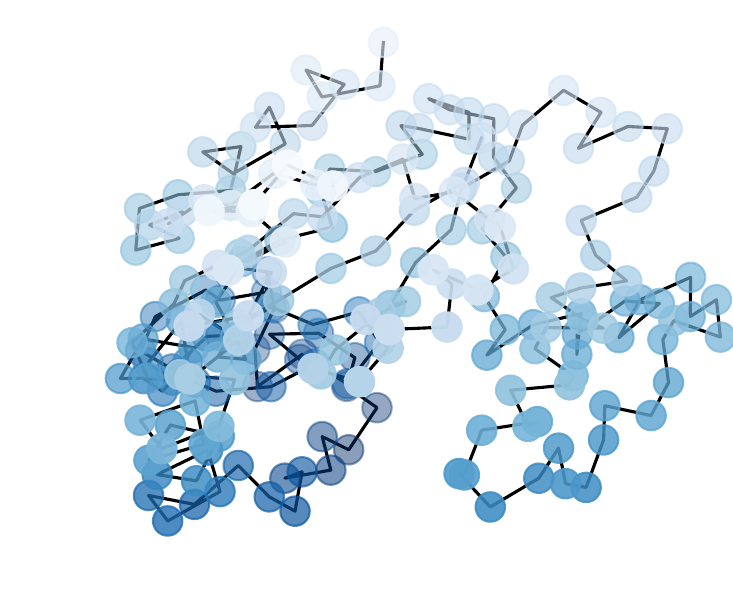}
         \caption{300 projection images.}
         \label{fig:better_reconstruction}
     \end{subfigure}
     \label{fig:results_3d_plot}
     \caption{Two reconstructions of the target conformation using a different number of projections. Note that this depicts the 3D spatial arrangements, i.e, elements in $\AtmModSpace$. By comparing with \cref{fig:target_bb}, it is clear that the number of projections improves the reconstruction.}
\end{figure}

To improve the reconstruction, we increase number of projection images to $300$ and proceed as in \cref{rem:double_proj} (but time with $300$ elements in $\DataSpace$ instead of two).
The random projections are selected by uniformly sampling 300 rotation matrices that each describe a change to coordinates in which the projection is computed as the projection on the $(x,y)$-plane. 

Keeping everything else as above, we recover the reconstruction in \cref{fig:better_reconstruction}.
In \cref{fig:results} we illustrate the deformation path of the 3D arrangements of the template, i.e., the deformation path in $\AtmModSpace$. 
Comparing the deformed template to the target conformation in \cref{fig:target_bb}, it is clear that this is a better reconstruction. 
We manage to capture the closed-to-open movement of the adylenate protein and decently compute a reconstruction of the final structure.
The matching is by no means perfect, but the large-scale movement is indeed captured and highlights that shape matching can be used for this purpose. 
\begin{figure}
    \centering
    \includegraphics[width = \linewidth]{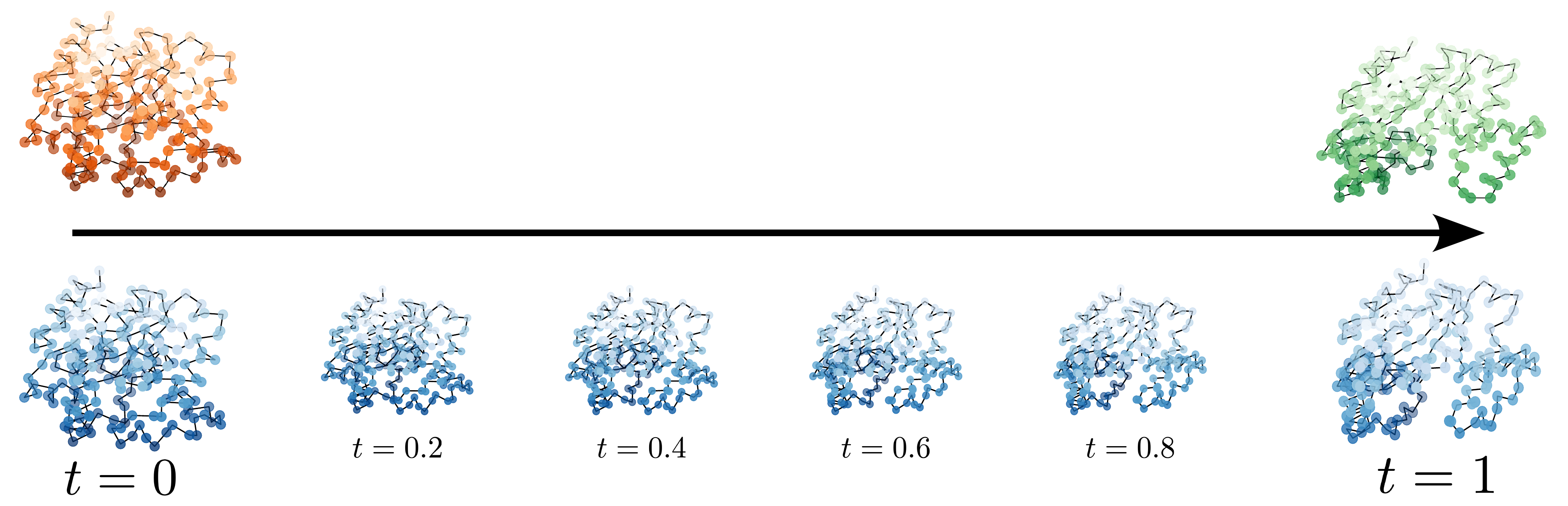}
    \caption{The results of applying the indirect shape matching algorithm to the template (orange structure). The evolution is shown at each included time grid point (blue structures). Compare with the final, not directly observed, frame (green structure). }
    \label{fig:results}
\end{figure}

To give a quantitative assessment of the dependence of the quality of the reconstruction on the number of projections, we use Procrustes analysis to compare the deformed point cloud to the reconstructed conformation. 
Procrustes analysis aligns points clouds by minimizing the \emph{Procrustes score}, which is a statistical measure of shape similarity. A lower score means that two compared point clouds are more similar \cite{Gower1975,Kendall1989}. 
We run, for $k = 2^1, 2^2, \dots, 2^9$ randomly chosen projections, the reconstruction procedure as described above, using an image noise standard deviation of $1.0$.

To give an idea of the uncertainty of the reconstruction, we repeat the procedure, with new noisy targets, $20$ times for each noise level. 
Then, the $20$ reconstructions are assessed against the ground truth, i.e., the target point cloud, by computing the average Procrustes score. 
This results in \cref{fig:proc-curve}, where the average Procrustes score is plotted against the number of projections in a log-log scale, together with the 90\% quantiles. 
\Cref{fig:proc-curve} confirms that indeed, more projections results in a lower Procrustes score. This indicates that the reconstructed protein conformation is closer to the target conformation in the sense of point clouds if more projections are used. 
\Cref{fig:proc-curve} indicates that the Procrustes score decreases as $\opStyle{O}(N_P^r)$, where $N_P$ is the number of projections. 
A linear regression fit estimates the slope $r$ as $-0.48$, indicating that $r \approx 0.5$.

Moreover, the uncertainty becomes smaller as the number of projections goes up. 
This translates in practice to the fact that if we just have a few randomly chosen projections, we need to be lucky and obtain good projection angles, whereas if we have more, we can be more certain that we have enough projection angles to obtain a good reconstruction.

\begin{figure}
    \centering
    \includegraphics[scale = 0.6]{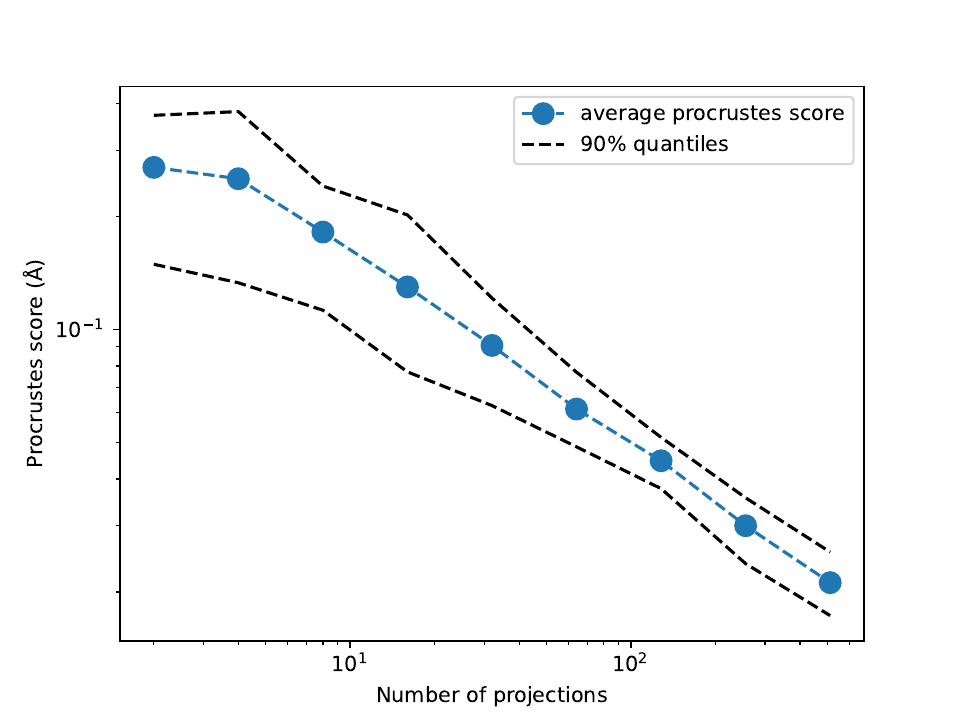}
    \caption{The average Procrustes score, in Ångströms, is plotted against the number of projections in a log-log scale. Note that the more projections that are used, the lower the Procrustes score. }
    \label{fig:proc-curve}
\end{figure}

We perform a similar experiment to investigate the effect of the noise. 
For a fixed number of projections, we run the reconstruction algorithm as above, but with varying image noise standard deviations. 
For $10$, $50$ and $100$ projections, we let the noise standard deviations take the values  $0.1,0.2,0.4,0.8,1.2$ and $2.4$, repeating the reconstruction  $M = 20$ times. 
The $20$ reconstructions are assessed against the ground truth by computing the Procrustes scores. 
This results in \cref{fig:proc-curve-noise}, where the average Procrustes score is plotted against the number of projections in a log-log scale, together with the 90\% quantiles.  
We observe that, as expected, more projections increase the reconstruction performance for higher noise standard deviations, but that the decrease in performance is qualitatively similar for $10, 50$ and $100$ projections. 
The number of projections further seems to influence how the width of the quantiles increases as the variance of the image noise increases, more projection images results in less wide quantiles, indicating that the uncertainty of the reconstruction is dependent on the number of projections as well as on the noise level. 

\begin{figure}
    \centering
    \includegraphics[scale = 0.6]{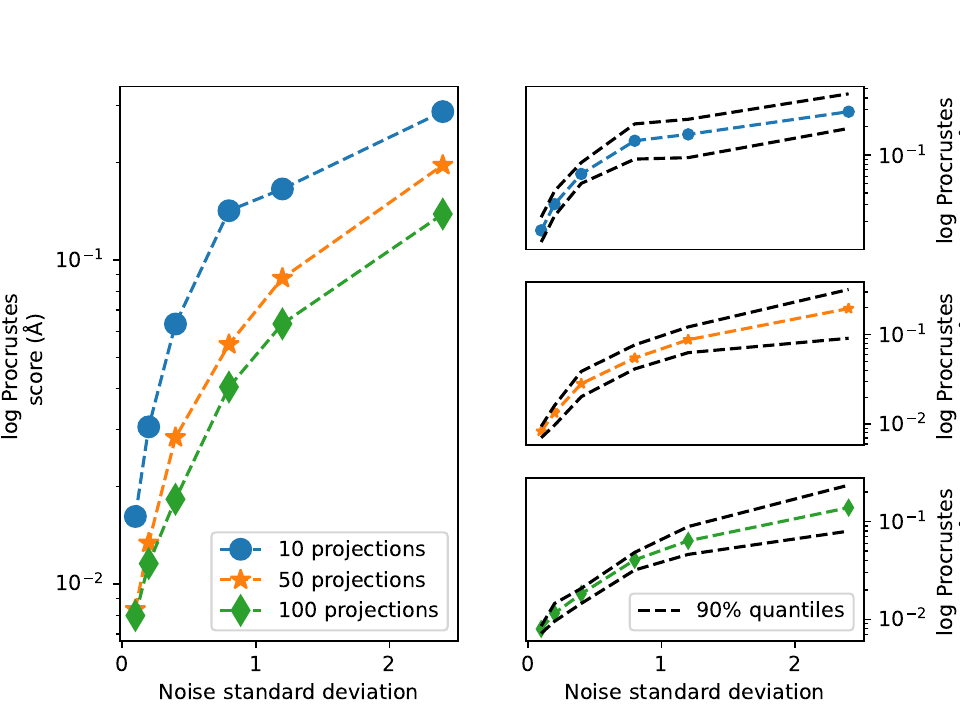}
    \caption{The average Procrustes score, in Ångströms, is plotted against the number of projections in a log-lin scale. Note that the robustness against noise increases with increasing number of projections. Further, the number of projections seems to influence the width of the quantiles, indicating that the uncertainty of the reconstruction is dependent on the number of projections as well as on the noise level. }
    \label{fig:proc-curve-noise}
\end{figure}

There are other ways than Procrustes analysis to quantitatively assess the reconstruction, see for instance \cite{Taylor2001}.
We intend to investigate the use of other protein similarity measures in future work.

\section{Conclusions and outlook}
\label{sec:outlook}

This work demonstrates that geodesic shape matching can be used for solving \cref{eq:Joint3DModelHomo}, which formalizes the task of reconstructing the 3D backbone conformations of a protein from noisy \ac{Cryo-SPA} data. 
The conformations of a protein backbone are modelled as deformations of a template.
The deformation is given by the action of a finite-dimensional group on a suitable set of deformable options, which in this case is the $\numPep$-fold direct product of $\SO(3)$.  
To find the deformation associated with shape matching, which is the conformation, we minimize the functional in \cref{eq:general_energy}.
For this, we presented two gradient-based algorithms, and we  illustrate the method with an example. The performance was assessed quantitatively using Procrustes analysis, as we varied the amount of data and its noise level.  
Results strongly suggest that geodesic shape matching can be used for reconstructing 3D backbone conformations of a protein from  noisy \ac{Cryo-SPA} data. 

The above approach can be extended in several directions, all motivated by solving the more challenging problem \cref{eq:Joint3DModelHetro} that arises in \ac{Cryo-SPA}. 
A straightforward extension is to incorporate pose estimation by combining the method presented in this paper with, for instance, existing pose estimation algorithms, like \cite{Diepeveen:2023aa}.
One could also have an intertwined approach where poses are estimated or marginalized as part of a refinement scheme.
Hence, we claim that one can drop the assumption in 
\cref{eq:Joint3DModelHomo} of known pose.

The flexibility of the shape matching framework can also be used for other extensions. 
In principle, with all else fixed, the geometric behavior of the method is only affected by the choice of regularization term. 
This opens up for several directions of future research. 

First, the regularization term in \cref{eq:RegFunc} is fixed up to the choice of the symmetric positive semi-definite operator $\mathbb{I}$ that is implicitly given by the relation \cref{eq:IOperator}.
As an example, the $i$:th component of $\mathbb{I}$ models the inertia of the $i$:th \ch{C_{$\alpha$}} atom.
Hence, one can in principle account for the influence of side-chains, since a long, heavy side-chain means that the protein rotates more easily in certain directions. 
Likewise, $\mathbb{I}$ could also account for  internal forces in the protein that are generated by electrical and chemical interactions. 

Second, our framework allows for selecting different data fidelity terms, i.e., different $\LossFunc_{\DataSpace}$ in \cref{eq:general_energy}. 
This makes it possible to account for chemical bonds and nonlinear interactions between atoms in the protein. 
In addition, one can incorporate an optimal-transport based similarity measure, which would allow for increased robustness to noise in data.

\revision{A third line of development concerns extending our shape-based indirect registration method to the heterogeneous particle setting, where the particle conformations in the sample are sampled from a continuum of possible conformations, and we aim to reconstruct a conformation for each image in the dataset.
Here one has the option to fit a single curve of deformations to all of the data (2D particle image), or have separate curves of deformations for each data point.
The second option amounts to applying the method presented in this work separately to each of the conformations represented in the sample.
Although less straightforward, the first option seems more attractive since it would make joint use of the data.
In order to make use of this approach, a potentially interesting line of research would be to investigate if the metric parametrised by $\mathbb{I}$ could be chosen in such a way that the conformational variability of the particle would be approximately described by a geodesic, which would motivate the use of $\RegFunc$ as a regulariser.
In this setting we would also need to account for both particle poses and matching data points to points on the curve.
One way of doing this is to view particle specific poses and the temporal ordering of conformational states as random variables and then to marginalise over these.
More precisely, we introduce the $\SE(3)$-valued random variables $\stpose_i\sim\mathrm{P}^i(\SE(3))$ that describe the pose (rotation and position) in space of each particle. 
Likewise, we also introduce $[0,1]$-valued random variables $\stt_i\sim\mathrm{P}^i([0,1])$ that represent the temporal ordering of the conformational state for the particle.
Many methods for 3D map reconstruction that handle unknown poses, like RELION \cite{Scheres:2012ab}, use an iterative refinement procedure that alternates between updating pose distributions and updating the 3D map while marginalising over poses. This could be extended by intertwining updates for the distributions of temporal orderings. 
Likewise, the joint reconstruction and model building step could include a marginalisation over poses and temporal orderings:
\begin{equation*}
     \hat{\mathbb{I}} \in \argmin_\mathbb{I} \biggl\{ \min_{\acurve}
  \sum_{i=1}^\numImgs \mathbb{E}_{\subalign{& \stt_i  \sim\mathrm{P}^i([0,1]) \\ & \stpose_i \sim\mathrm{P}^i(\SE(3))}}\left[\LossFunc_{\DataSpace}\Bigl((\ForwardOp\circ \GroupAction)\bigl(\gcurve^{\acurve}(\stt_i),\stpose_i.\template\bigr),\target_i\Bigr)\right]
    + \lambda \RegFunc(\acurve,\mathbb I) \biggr\}.
\end{equation*}
Given $\hat{\mathbb{I}}$, we can then compute a minimizer $\hat{\acurve}$ that in turn yields a curve $\gcurve^{\hat{\acurve}}$ which describes the conformational variability of the biomolecule.
Here, we made the dependence of $\RegFunc$ on $\mathbb{I}$ explicit, and $\pose_i.\template \in \ShapeSpace$ for some $\pose_i \in \SE(3)$ is the rotated and translated template $\template$ given by acting with the roto-translation group $\SE(3)$ on the shape space $\ShapeSpace$.
The above formulation is in its current form not feasible for implementation, but it provides a starting point to adapt our approach for joint reconstruction and model building to the heterogenous particle setting.}

A final direction of future research is to forgo the geodesic shape matching framework altogether, and instead work directly with gradient flows on the space of protein backbones chains using an adaptation of the method described in \cite{BaKaMo2022}.  
Such an approach should result in a substantial reduction of complexity without impairing upon performance. 

\newpage
\appendix
\section{Proof of \cref{th:gradient_general}}
\label{app:A}
We here prove \cref{th:gradient_general} in the case when $\LieGroup$ is a matrix Lie group. The proof is an adaptation of the proof of \cite[Theorem~2.1]{beg2005}. 
Hence, we start with the following finite-dimensional variant of \cite[Lemma~2.1]{beg2005} that computes the variation of the $\gcurve^{\acurve} \colon [0,1] \to \LieGroup$ where $\gcurve^{\acurve}$ solves the flow equation \cref{eq:flow2} under the perturbation of the curve $\acurve \colon [0,1] \to \LieAlgebra$.
\begin{lemma}\label{th:variation}
  Let $\acurve$ be a smooth curve in the Lie algebra $\LieAlgebra$ of the matrix Lie group $\LieGroup$ and let $\gcurve^{\acurve} \colon [0,1] \to \LieGroup$ denote the curve that solves the flow equation $\dot \gcurve^{\acurve} = \acurve(t)\gcurve^{\acurve}(t)$, $\gcurve^{\acurve}(0)=e$. 
  The (infinitesimal) variation of the $\LieGroup$-curve $\gcurve^{\acurve}$, denoted by $\delta_{\dacurve} \gcurve$, as the $\LieAlgebra$-curve $\acurve$ is perturbed with $\dacurve \in \LpSpace^2\bigl([0,1],\LieAlgebra\bigr)$ is then given by
  \begin{equation}\label{eq:var}
    \delta_{\dacurve} \gcurve^{\acurve}(t) = dL_{\gcurve^{\acurve}(t)} \biggl( \int_0^t \AdjRep_{(\gcurve^{\acurve}(s))^{-1}}\bigl(\dacurve(s)\bigr) \dint s
    \biggr)
  \end{equation}
  where the adjoint representation $\AdjRep_\gelem \colon \LieAlgebra \to \LieAlgebra$ for fixed $\gelem \in \LieGroup$ is given by $\AdjRep_{\gelem}(\aelem) := \gelem \aelem \gelem^{-1}$.
\end{lemma}
\begin{proof}
$\LieGroup$ is a matrix Lie group, so the curve $\gcurve^{\acurve} \colon [0,1] \to \LieGroup$ with $\acurve \in \LpSpace^2\bigl([0,1],\LieAlgebra\bigr)$ given by in \cref{eq:flow2} takes the form
\[
    \dot \gcurve^{\acurve}(t) = \acurve(t)\gcurve(t), \, \gcurve(0) = I \in \LieGroup.
\]
Let us take an admissible first order variation of $\acurve$ in the $\dacurve$-direction, which is fixed, by considering  
\[ \acurve_{\epsilon}(t) 
     := \acurve(t) + \epsilon \dacurve(t).
\]
This variation induces a variation in $\gcurve^{\acurve}$ along the $\dacurve$-direction that can be expressed as 
\begin{equation}\label{eq:VarCurve}
    \dot \gcurve^{\acurve_{\epsilon}}(t) = \bigl(
      \acurve(t) + \epsilon \dacurve(t) \bigr) 
      \gcurve^{\acurve_{\epsilon}}(t).  
\end{equation}
Furthermore, we also get that the infinitesimal variation in $\gcurve^{\acurve}$ along the $\dacurve$-direction is given by 
\begin{equation}\label{eq:VarCurve2}
\delta_{\dacurve} \gcurve^{\acurve}(t) 
  = \dfrac{d}{d \epsilon}
     \gcurve^{\acurve_{\epsilon}}(t)\Bigl\vert_{\epsilon=0}.
\end{equation}
Differentiate \cref{eq:VarCurve} with respect to $\epsilon$ and evaluate it at $\epsilon = 0$. Combining this with \cref{eq:VarCurve2} yields the following differential equation:
\begin{equation}\label{eq:ode}
  \partial_t \delta_{\dacurve} \gcurve^{\acurve}(t) = \acurve(t) \delta_{\dacurve} \gcurve^{\acurve}(t)+ \dacurve(t) \gcurve^{\acurve}(t).
\end{equation}  
The variation of $\gcurve^{\acurve}(t)$ along the $\dacurve$-direction is thus given by a solution of \cref{eq:ode}, which in turn can be expressed as
\begin{equation}\label{eq:variation}
  \delta_{\dacurve} \gcurve^{\acurve}(t) = \gcurve^{\acurve}(t) \int_0^t (\gcurve^{\acurve}(s))^{-1} \dacurve(s)\gcurve(s)\dint s. 
\end{equation}
Indeed, differentiating \cref{eq:variation} with respect to $t$ yields \cref{eq:ode}.
To complete the proof, we next insert the definition of the adjoint representation. This gives us 
\[ 
\bigl(\gcurve^{\acurve}(s)\bigl)^{-1} \dacurve(s)\gcurve^{\acurve}(s) =  \AdjRep_{(\gcurve^{\acurve}(s))^{-1}} 
\bigl(\dacurve(s)\bigr).
\]
Finally, we note that, 
\[ 
\int_0^t \Bigl(
  \gcurve^{\acurve}(s)\bigr)^{-1} 
  \dacurve(s)
  \gcurve^{\acurve}(s)
  \Bigr)
  \dint s 
  \in \LieAlgebra
  \quad\text{for any $t \in [0,1]$,}
\]
so 
\[
\delta_{\dacurve} \gcurve^{\acurve}(t) 
= \gcurve^{\acurve}(t)\int_0^t 
\bigl(\gcurve^{\acurve}(s)\bigr)^{-1} \dacurve(s)
\gcurve^{\acurve}(s)
\dint s 
= dL_{\gcurve(t)}\biggl( 
  \int_0^t \AdjRep_{(\gcurve^{\acurve}(s))^{-1}} \bigl( \dacurve(s) \bigr) \dint s
  \biggr). 
\]
This completes the proof. 
\end{proof}
\begin{proof}[Proof of \cref{th:gradient_general}]
The Gâteaux derivative of $\EnergyFunc \colon \ContSpace^{\infty}\bigl([0,1],\LieAlgebra\bigr) \to \Real$ (matching energy) in \cref{eq:general_energy} in the direction $\dacurve\in \LpSpace^2\bigl([0,1],\LieAlgebra\bigr)$ is defined as 
\begin{equation}\label{eq:GDiffE}
\delta_{\dacurve} \EnergyFunc(\acurve) :=\lim_{\varepsilon\to 0}\frac{\EnergyFunc(\acurve+\varepsilon \dacurve)-\EnergyFunc(\acurve)}{\varepsilon}.    
\end{equation}
If we define $\residual_{\template,\target} \in \DataSpace$ as in \cref{eq:Residual}, then we can re-write \cref{eq:GDiffE} as
\[
\delta_{\dacurve} \EnergyFunc(\acurve)
  =\int_0^1 \bigl\langle \acurve(t),\dacurve(t) \bigr\rangle_{\LieAlgebra} \dint t
  + 
  \biggl\langle 
  \residual_{\template,\target},
  d\bigl( \ForwardOp \circ \GroupAction(\Cdot,\template)\bigr)_{\gcurve^{\acurve}(1)} 
  \bigl(\delta_{\dacurve}\gcurve^{\acurve}(1)\bigr)
\biggr\rangle_{\DataSpace}.
\]
The proof is now completed by using \cref{th:variation} to rewrite the second term in the above expression.
To see this, note first that
\begin{multline*}
\biggl\langle 
  \residual_{\template,\target},
  d\bigl( \ForwardOp \circ \GroupAction(\Cdot,\template)\bigr)_{\gcurve^{\acurve}(1)} 
  \bigl(\delta_{\dacurve}\gcurve^{\acurve}(1)\bigr)
\biggr\rangle_{\DataSpace}
\\ \shoveleft{\qquad
=
\biggl\langle 
\residual_{\template,\target}, 
\Bigl( d\bigl(\ForwardOp\circ\GroupAction(\Cdot,\template)\bigr)_{\gcurve^{\acurve}(1)} \circ
d(L_{\gcurve^{\acurve}(1)})_e \Bigr) \Bigl(\int_0^1 \AdjRep_{(\gcurve^{\acurve}(t))^{-1}} \bigl(\dacurve(t)\bigr) \dint t \Bigr)
\biggr\rangle_{\DataSpace}
}
\\ \shoveleft{\qquad
=
\biggl\langle 
\residual_{\template,\target}, 
\varpi\Bigl(\int_0^1\AdjRep_{(\gcurve^{\acurve}(t))^{-1}} \bigl(\dacurve(t)\bigr) \dint t \Bigr)
\biggr\rangle_{\DataSpace}
}
=
\biggl\langle 
\varpi^*(\residual_{\template,\target}), 
\int_0^1\AdjRep_{(\gcurve^{\acurve}(t))^{-1}} \bigl(\dacurve(t)\bigr) \dint t
\biggr\rangle_{\DataSpace}
\end{multline*}
where $\varpi \colon \LieAlgebra \to \DataSpace$ is the linear mapping 
\[ \varpi := d\bigl(\ForwardOp\circ\GroupAction(\Cdot,\template)\circ L_{\gcurve^{\acurve}(1)}\bigr)_e.
\]
The final step is to exchange the order with which we take the $t$-integral and the $\DataSpace$-inner product:
\begin{multline*}
\biggl\langle 
  \residual_{\template,\target},
  d\bigl( \ForwardOp \circ \GroupAction(\Cdot,\template)\bigr)_{\gcurve^{\acurve}(1)} 
  \bigl(\delta_{\dacurve}\gcurve^{\acurve}(1)\bigr)
\biggr\rangle_{\DataSpace}
=
\biggl\langle 
\varpi^*(\residual_{\template,\target}), 
\int_0^1\AdjRep_{(\gcurve^{\acurve}(t))^{-1}} \dacurve(t) \dint t
\biggr\rangle_{\DataSpace}
\\
=\int_0^1\Bigl\langle 
  \varpi^*(\residual_{\template,\target}), 
  \AdjRep_{(\gcurve^{\acurve}(t))^{-1}} \bigl( \dacurve(t) \bigr)
  \Bigr\rangle_{\LieAlgebra} \dint t
\\
=\int_0^1\Bigl\langle 
\bigl(\AdjRep_{(\gcurve^{\acurve}(t))^{-1}}\bigr)^*
\bigl( \varpi^*(\residual_{\template,\target}) \bigr), 
\dacurve(t)
\Bigr\rangle_\LieAlgebra dt.
\end{multline*}
\end{proof}

\section{Gradient calculations}
\label{app:B}
Here, we detail the results of applying the setting of \cref{sec:protein_specifics} to  \cref{th:gradient_general} to obtain a closed-form expression for the gradients used in computing the matching in \cref{sec:experiment}. 

\begin{corollary}\label{cor:gradient_so3}
Consider the Lie group $\LieGroup = (\SO(3))^\numPep$ that acts on deformable objects in the shape space $\ShapeSpace= \Real^{3 \times \numPep}$. Moreover, assume observed data resides in $\DataSpace = \LpSpace^2(\Real^2)$.
Next, consider the functional $\EnergyFunc \colon \ContSpace^{\infty}\bigl([0,1],\LieAlgebra\bigr) \to \Real$ defined as in \cref{eq:general_energy} where $\LieAlgebra=\bigl(\SOLieAlgebra(3)\bigr)^\numPep$ is the Lie algebra to $\LieGroup$.
The target in $\DataSpace$ is denoted by $\target$
and the template in $\ShapeSpace$ by $\template$. 
Then, the gradient of $\EnergyFunc$ is given as
\begin{equation}
    \nabla_{\acurve} \EnergyFunc(\acurve) =\lambda \acurve + \Bigr( \AdjRep_{(\gcurve^{\acurve})^{-1}}^* \circ \,d\bigl(\GroupAction(\cdot, \template) \circ L_{\gcurve^{\acurve}(t)}\bigr)_e^* \circ \opStyle{M}^* \circ \pi^* \circ (d\AtmToMapOpPlane)_{q}^*\Bigr)(\eta)
\end{equation}
where  
\begin{equation*}
    \eta :=  \bigl(\AtmToMapOpPlane \circ \pi \circ \opStyle{M}\bigr)\Bigl(\GroupAction\bigl(\gcurve^{\acurve}(1),\template\bigr)\Bigr) -\target
    \quad\text{and}\quad
q := \bigl(\pi \circ \opStyle{M}\bigr)\Bigl(\GroupAction\bigl(\gcurve^{\acurve}(1),\template\bigr)\Bigr)
\end{equation*}  
and the map
\[ \target \mapsto \Bigl(\AdjRep_{(\gcurve^{\acurve}(t))^{-1}}^* \circ\, d\bigl(\GroupAction(\cdot,\template)\circ L_{\gcurve^{\acurve}(t)}\bigr)_e^* \circ \opStyle{M}^* \circ \pi^* \circ \bigl(d(\AtmToMapOpPlane)_q\bigr)^* \Bigr)(\target)
%\quad\text{for $f \in \DataSpace$}
\]
is computed by the following operators: 
\begin{enumerate}
\item $\bigl((d\AtmToMapOpPlane)_q \bigr)^* \colon \DataSpace \to \Real^{2 \times \numPep}$ that is defined as
\[
      \bigl(d(\AtmToMapOpPlane)_q\bigr)^*(\target) := \bigl(
      s_1(\target\ast \nabla g_{\sigma_1})(q_1), \dots,  s_{\numPep}(\target\ast \nabla g_{\sigma_{\numPep}})(q_{\numPep}) \bigr) = \projbbpos \in \Real^{2 \times \numPep},
\]      
\item $\pi^* \colon \Real^{2 \times \numPep} \to \AtmModSpace$ that is defined as
\[
      \pi^*(\projbbpos) := 
      \bigr( \Lambda \projbbpos_1, \Lambda \projbbpos_2,  \ldots, \Lambda \projbbpos_{\numPep}
      \bigr) = \bbpos \in \AtmModSpace
      \quad\text{where}\quad
      \Lambda =  \begin{pmatrix}
        1 & 0 \\
        0 & 1 \\
        0 & 0
    \end{pmatrix},
\]    
\item $\opStyle{M}^* \colon \AtmModSpace \to \ShapeSpace$ that is defined as
\[ 
\opStyle{M}^*(\bbpos) :=
  \biggl(
  \sum_{i=1}^\numPep \bbpos_i, \ldots,  \bbpos_{\numPep-1}+\bbpos_{\numPep},  \bbpos_{\numPep}\biggr) 
  = \relbbpos \in \ShapeSpace,
\]
\item $\Bigl(d\bigl( \GroupAction(\cdot,\template)\circ L_{\gcurve^{\acurve}(t)}\bigr)_e\Bigr)^* \colon \ShapeSpace \to \LieAlgebra$ that is defined as
\[ 
\Bigl(d\bigl(\GroupAction(\cdot,\template)\circ L_{\gcurve^{\acurve}(t)}\bigr)_e\Bigr)^*(\relbbpos) :=
  \biggl(
    (-\mathbb{I}_1^{-1} \circ \Pi)\Bigl(\template_1 \relbbpos_1^{\top} \bigl(\gcurve^{\acurve}(t)\bigr)_1\Bigr) 
    , \dots, 
    (-\mathbb{I}_{\numPep}^{-1} \circ \Pi)\Bigl(\template_{\numPep} \relbbpos_{\numPep}^{\top} \bigl(\gcurve^{\acurve}(t)\bigr)_{\numPep} \Bigr)
  \biggr)
\]  
where  $\Pi \colon \GLLieAlgebra(3,\Real)\to \SOLieAlgebra(3)$ denotes the orthogonal projection relative to the Frobenius inner product onto the Lie algebra $ \SOLieAlgebra(3)$. Explicitly, $\Pi(W)=(W-W^{\top})/2$ for any matrix $W\in \GLLieAlgebra(3,\Real)$.
\item $\AdjRep_{\gcurve^{\acurve}(t)}^* \colon \LieAlgebra \to \LieAlgebra$ that is defined as
\[ \AdjRep_{\gcurve^{\acurve}(t)}^*(\aelem) := \Bigl( (\mathbb{I}_1^{-1} \circ \AdjRep_{(\gcurve^{\acurve}(t))_1^{\top}}\mathbb{I}_1)(\aelem_1), \dots,(\mathbb{I}_{\numPep}^{-1} \circ \AdjRep_{(\gcurve^{\acurve}(t))_{\numPep}^{\top}} \circ \mathbb{I}_{\numPep})( \aelem_{\numPep)}
\Bigr)
\quad\text{for $\aelem \in \LieAlgebra$.}
\]
\end{enumerate}
\end{corollary}
\begin{proof}
The proof amounts to computing the expressions of the adjoint mappings, which we compute one by one.
We first calculate the adjoint of the map $\AdjRep_\rotmat \colon\bigl(\SOLieAlgebra(3)\bigr)^\numPep\to\bigl(\SOLieAlgebra(3)\bigr)^\numPep$ for any given $\rotmat \in (\SO(3))^\numPep$. 
    
Since the map $\AdjRep_\rotmat$ acts component wise, we assume for the moment in our calculations that $\numPep=1$. It holds that
\begin{equation}
    \inpro[]{\AdjRep_\rotmat(\aelemother)}{\aelem}=\inpro[F]{\AdjRep_\rotmat(\aelemother)}{\mathbb{I}(\aelem)}=-\traceop\bigl(\rotmat \aelemother\rotmat^{\top}\mathbb{I}(\aelem)\bigr)=-\traceop(\aelemother\rotmat^{\top}\mathbb{I}(\aelem)\rotmat)=\inpro[F]{\aelemother}{(\AdjRep_{\rotmat^{\top}} \circ\mathbb{I})(\aelem)},
\end{equation}
for an arbitrary $\aelem \in \SOLieAlgebra(3)$,
showing that
\begin{equation}
    \AdjRep_{\rotmat}^*(\aelem)=(\mathbb{I}^{-1} \circ \AdjRep_{\rotmat^{\top}} \circ \mathbb{I})(\aelem).
\end{equation}
In the general case $\numPep\geq1$, the formula is given by
\begin{equation}
    \AdjRep_{\rotmat}^*(\aelem)=
    \Bigl((\mathbb{I}_1^{-1} \circ \AdjRep_{\rho_1^{\top}} \circ \mathbb{I}_1)(\aelem_1), \dots, (\mathbb{I}_{\numPep}^{-1} \circ \AdjRep_{\rotmat_{\numPep}^{\top}} \circ \mathbb{I}_{\numPep})(\aelem_{\numPep})
    \Bigr)
    \quad\text{for arbitrary $\aelem \in \LieAlgebra$.}
\end{equation}
Next, we calculate the adjoint of the map $d(\GroupAction(\cdot,\relbbpos)\circ L_\rotmat)_e:\LieAlgebra\to \ShapeSpace$ for a fixed $\relbbpos\in \ShapeSpace$ and $\rotmat\in \LieGroup$. 
Just as  when we calculated the adjoint of the adjoint representation, we first restrict ourselves to the case $\numPep=1$. Differentiating the relation $(\GroupAction(\cdot,\relbbpos)\circ L_\rotmat)(T)=\rotmat T \relbbpos$ with respect to an arbitrary $T\in\SO(3)$ gives $d(\GroupAction(\cdot,\relbbpos)\circ L_\rotmat)_e \aelem=\rotmat \aelem \relbbpos$, where $\aelem$ is an element of the Lie algebra $\SOLieAlgebra(3)$. 
For  an arbitrary vector $b\in\Real^3$, we have that
\begin{multline}
        \inpro[]{d(\GroupAction(\cdot,\relbbpos)\circ L_\rotmat)_e (\aelem)}{b}
        =\inpro[]{\rotmat \aelem \relbbpos}{b}
        =b^{\top}\rotmat \aelem \relbbpos
        =\traceop(b^{\top}\rotmat \aelem \relbbpos)\\
        =\traceop(\relbbpos b^{\top} \rotmat \aelem)
        =\traceop(\Pi(\relbbpos b^{\top}\rotmat)\aelem)
        =-\inpro[F]{\Pi(\relbbpos b^{\top}\rotmat)}{\aelem},
\end{multline}
which shows that
\begin{equation}
    (d(\GroupAction(\cdot, \relbbpos)\circ L_\rotmat)_e)^*(b)=-(\mathbb{I}^{-1} \circ \Pi)(\relbbpos b^{\top}\rotmat),
\end{equation}
and in the general case for $\numPep\geq 1$ this becomes
\begin{align*}    \bigl(d(\GroupAction(\cdot,\relbbpos)\circ L_\rotmat)_e\bigr)^*(b)=
    \Bigl( (-\mathbb{I}_1^{-1} \circ \Pi)(\relbbpos_1b_1^{\top}\rotmat_1), \dots, (-\mathbb{I}_{\numPep}^{-1}\circ \Pi)(\relbbpos_\numPep b_\numPep^{\top}\rotmat_\numPep)
    \Bigr).
\end{align*}
 We turn our attention to the adjoints of the forward model. Firstly, the adjoint of the map $\opStyle{M}:\ShapeSpace\to \AtmModSpace$ is obtained from the relations
\begin{equation*}
        \inpro[]{\opStyle{M}(\relbbpos)}{\bbpos}
        =\sum_{i=1}^\numPep\inpro[]{\sum_{j=1}^i\relbbpos_j}{\bbpos_i}
        =\sum_{i=1}^\numPep\sum_{j=1}^i\inpro[]{\relbbpos_j}{\bbpos_i}
        =\sum_{j=1}^\numPep\sum_{i=j}^\numPep\inpro[]{\relbbpos_j}{\bbpos_i}
        =\sum_{j=1}^\numPep \inpro[]{\relbbpos_j}{\sum_{i=j}^\numPep \bbpos_i},
\end{equation*}
showing that
\begin{equation}
    \opStyle{M}^*(\bbpos)=\left(\sum_{i=1}^\numPep \bbpos_i,\dots,\bbpos_{\numPep-1}+\bbpos_{\numPep},\bbpos_{\numPep}\right).
\end{equation}
The adjoint of the projection map $(x,y,z)\mapsto (x,y)$ is given by the inclusion $(x,y)\mapsto(x,y,0)$. Thus, the  adjoint of the map $\pi:\AtmModSpace \to\Real^{2\times \numPep}$, which is given by a projection on each component, is obtained by including each of the components of $\Real^{2\times \numPep}$ into $\Real^3$. The adjoint of the projection map $(x,y,z)\mapsto (x,y)$ is given by the matrix-vector product 
\[
    \Lambda \begin{pmatrix}
         x\\
         y 
    \end{pmatrix} = \begin{pmatrix}
         x\\ 
         y \\
         0
    \end{pmatrix} 
\quad\text{where}\quad
   \Lambda :=  \begin{pmatrix}
        1 & 0 \\
        0 & 1 \\
        0 & 0
    \end{pmatrix}.
\]
The extension to $\Real^{2 \times \numPep}$ is the linear map $\Lambda \oplus \ldots \oplus \Lambda \colon \Real^{2\times \numPep} \to \AtmModSpace$ obtained by applying $\Lambda$ to each component.

Finally, we calculate the adjoint of $d(\AtmToMapOpPlane)_q \colon \Real^{2\times \numPep}\to \DataSpace$. Differentiating $\AtmToMapOpPlane$ at an arbitrary point $q\in \Real^{2\times \numPep}$ gives 
\[ d(\AtmToMapOpPlane)_q(q')=\sum_{i=1}^\numPep\langle -q'_i,s_i\tau_{q_i}\nabla \gauss_{\sigma_i}\rangle = \sum_{i=1}^\numPep\langle q'_i,-s_i\tau_{q_i}\nabla \gauss_{\sigma_i}\rangle
\quad\text{for $q'\in\Real^{2\times \numPep}$.}
\]
Thus, if $\data \in \DataSpace$, we have
\begin{align*}
    \bigl\langle d(\AtmToMapOpPlane)_q(q'),\data \bigr\rangle_{\DataSpace}=\sum_{i=1}^\numPep \left\langle q'_i, -s_i\int_{\Real^2} \data \tau_{q_i}\nabla \gauss_{\sigma_i} \right\rangle=\sum_{i=1}^\numPep \bigl\langle
    q'_i,s_i(\data \ast \nabla \gauss_{\sigma_i})(q_i)\bigr\rangle.
\end{align*}
The second equality above follows from the following identity;
\[
-\int_{\Real^2} \data \tau_{q_i}\nabla \gauss_{\sigma_i} =(\data \ast \nabla \gauss_{\sigma_i})(q_i)
\]
which holds since the gradient of an even function (in our case: the Gaussian probability density function) is odd.
Furthermore, the integral above acts component wise on the integrand, and $(\data\ast \nabla \gauss_{\sigma_i})(q_i)$ is the convolution of $\data$ with $\nabla \gauss_{\sigma_i}$ evaluated at the point $q_i$.
\end{proof}

\section{An existence result}
This section formulates the minimization of \cref{eq:general_energy} over $\LpSpace^2$ curves in the Lie algebra $\LieAlgebra$ following \cite{Younes2010}.
The existence of a minimizer can be proven in the same way as in the proof of \cite[Theorem~7.4]{Chen:2018aa}.
Most of the results stated here are special cases of more general results found in \cite{Younes2010}, but the proofs in our specific setting are simpler.

Let $\MatrixGroup$ be the space of complex $n\times n$-matrices equipped with the operator norm $\|\cdot\|$.
Note that $\MatrixGroup$ is finite dimensional, so the operator norm topology is the same as the one induced by an inner product, for example.
The operator norm makes the proofs in this appendix easier, which is why we have chosen this norm.
The Banach space $\LpSpace^2([0,1],\MatrixGroup)$ is defined as the space of all measurable curves $\acurve \colon [0,1]\to \MatrixGroup$ such that 
\[ \int_0^1 \bigl\|\acurve(s) \bigr\|_{\MatrixGroup}^2 \dint s<\infty. \]

Given a curve $\acurve\in \LpSpace^2([0,1],\MatrixGroup)$, we say that a continuous curve $\gcurve \colon [0,1]\to \MatrixGroup$ is a solution to the differential equation $\dot\gcurve=\acurve \gcurve$ with initial condition $\gcurve(0)=e$ (special case of \cref{eq:flow2} when $\LieGroup\subset \MatrixGroup$), in which $e$ is the identity matrix, if 
\[ \gcurve(t)=e+\int_0^t \acurve(s)\gcurve(s) \dint s. \]
The existence and uniqueness of solutions to the differential equation is guaranteed by \cite[Corollary~C.7]{Younes2010}.
We shall often write $\gcurve^{\acurve}$ for this unique solution to emphasize the dependence on $\acurve$.

We now show that the map $\acurve \mapsto \gcurve^{\acurve}(t)$ is weakly continuous for any $t\in [0,1]$, which in principle follows from \cite[Section~7.2.4]{Younes2010}.
\begin{proposition}\label{prp:SolWeakCont}
  Let $\MatrixGroup$ be the space of complex $n\times n$-matrices.
  For any fixed $t\in[0,1]$, the map $\LpSpace^2([0,1],\MatrixGroup) \ni \acurve \mapsto \gcurve^{\acurve}(t) \in  \MatrixGroup$, is weakly continuous on the closed unit ball of $\LpSpace^2([0,1],\MatrixGroup)$.
\end{proposition}
\begin{proof}
    Take $\acurve_1,\acurve_2\in \LpSpace^2([0,1],\MatrixGroup)$ with $\|\acurve_2\|\leq 1$.
    Write $\gcurve_i=\gcurve^{\acurve_i}$ for ease of notation.
    By definition, it holds that
    \begin{align*}
        \gcurve_i(t)=e+\int_0^t \acurve_i(s)\gcurve_i(s)ds.
    \end{align*}
    Then
    \begin{multline*}
    \|\gcurve_2(t)-\gcurve_1(t)\|
    \leq\left\|\int_0^t(\acurve_1(s)-\acurve_2(s))\gcurve_1(s)ds\right\|+\int_0^t\|\acurve_2(s)(\gcurve_1(s)-\gcurve_2(s))\|ds\\
    =\|\Lambda_t(\acurve_1-\acurve_2)\|+\int_0^t\|\acurve_2(s)(\gcurve_1(s)-\gcurve_2(s))\|ds.
    \end{multline*}
    where $\Lambda_t(\acurve)=\int_0^t\acurve(s)\gcurve_1(s)ds$ defines a bounded linear map $\Lambda_t:\LpSpace^2([0,1],\MatrixGroup)\to \MatrixGroup$.
    From Grönwall's lemma and the assumption $\|\acurve_2\|\leq1$, it follows that
    \begin{align*}
        \|\gcurve_2(t)-\gcurve_1(t)\|\leq \|\Lambda_t(\acurve_1-\acurve_2)\|e^{\int_0^t\|\acurve_2(s)\|ds}\leq\|\Lambda_t(\acurve_1-\acurve_2)\|e^{\sqrt{t}}.
    \end{align*}
    
    Since the map $\Lambda_t$ is bounded and has finite dimensional codomain, it is weakly continuous (note that boundedness implies continuity with respect to the weak topology on both domain and codomain, and since the codomain is finite dimensional the map is also continuous with respect to the weak topology on the domain and the norm topology on the codomain).
    Thus, for any $\varepsilon>0$ there is a weakly open neighborhood $U$ of $\acurve_1$ such that  $\|\gcurve_1(t)-\gcurve_2(t)\|<\varepsilon$ if $\acurve_2\in U$.
    It follows that the restriction of $\acurve\mapsto \gcurve^{\acurve}(t)$ to the closed unit ball is weakly continuous.
\end{proof}

Our next results show that if $\acurve$ lies in a Lie algebra of a matrix Lie group that is closed in $\MatrixGroup$ (such as $\SO(3)$), then the corresponding solution $\gcurve^{\acurve}$ lies in the group.
\begin{proposition}
    Suppose $\PoseGroup\subset \MatrixGroup$ is a closed matrix Lie group. If $\acurve\in \LpSpace^2([0,1],\LieAlgebra)$, then the unique continuous solution $\gcurve^{\acurve}:[0,1]\to \MatrixGroup$ to $\dot\gcurve=\acurve\gcurve$ with initial condition $\gcurve(0)=e$ is a curve in $\PoseGroup$.
\end{proposition}
\begin{proof}
    In the case that $\acurve$ is smooth, we have that there exists for each $p\in \PoseGroup$ a unique integral curve $\gcurve \colon J\to \PoseGroup$ for some open interval $J$ containing 0, of the smooth time-dependent vector field $W(t,g)=\acurve(t)g$ on $\PoseGroup$ starting at $p$.
    Note that $W$ extends to a smooth vector field on $\MatrixGroup$, so any integral curve of $W$ in $\MatrixGroup$ that starts at a point of $\PoseGroup$ stays in $\PoseGroup$, at least for small enough times, by uniqueness of integral curves.

    By the definition of $\gcurve^{\acurve}$ and an induction argument, $\gcurve^{\acurve}$ is smooth and thus is a smooth integral curve of $W$.
    
    By the above, the set $A$ of points in $[0,1]$ that $\gcurve^{\acurve}$ maps into $\PoseGroup$ is open.
    Since $A=(\gcurve^{\acurve})^{-1}(\PoseGroup)$ the set is closed by continuity.
    Since $0\in A$, we conclude that $A=[0,1]$.

    In general, if $\acurve\in \LpSpace^2([0,1],\LieAlgebra)$, there exists a sequence of smooth curves $(\acurve_k)$ that converges to $\acurve$ in $\LpSpace^2([0,1],\LieAlgebra)$.
    Since the map $\acurve\mapsto\gcurve^{\acurve}(t)$ for any fixed $t\in[0,1]$ is weakly continuous, as we saw in the previous proposition, it follows that every point $\gcurve^{\acurve}(t)$ lies arbitrarily close to $\PoseGroup$.
    As $\PoseGroup$ is closed, we must have $\gcurve^{\acurve}(t)\in \PoseGroup$.
\end{proof}
The above proposition shows that $\acurve\mapsto\gcurve^{\acurve}(1)$ maps $\LpSpace^2([0,1],\LieAlgebra)$ into $\PoseGroup$.
This result will be used implicitly in what follows.

Let $\ShapeSpace$ and $\DataSpace$ be topological spaces, and assume that the group $\PoseGroup$ acts continuously on $\ShapeSpace$ via $\GroupAction\colon\LieGroup\times\ShapeSpace\to\ShapeSpace$.
For a map $\opStyle{D} \colon \ShapeSpace \times \DataSpace \to[0,\infty)$, a constant $\lambda>0$, and fixed elements $\template\in \ShapeSpace$ and $\target\in \DataSpace$ we shall consider the functional
\begin{align*}
    \EnergyFunc(\acurve)=\opStyle{D}(\GroupAction(\gcurve^{\acurve}(1),\template) ,\target)+\lambda\|\acurve\|_2^2,
\end{align*}
which maps $\LpSpace^2([0,1],\LieAlgebra)$ into $[0,\infty)$.
We can now prove the existence of a minimizer of \cref{eq:general_energy}, using the same method as in the proof of Theorem~7.4 in \cite{Chen:2018aa}.
\begin{theorem}
    If $\opStyle{D}(\,\cdot\,,\target)$ is lower semi-continuous on $\ShapeSpace$, then there exists an element $\tilde\acurve\in \LpSpace^2([0,1],\LieAlgebra)$ such that 
    \[ \EnergyFunc(\tilde\acurve)=\inf_{\acurve\in \LpSpace^2([0,1],\LieAlgebra)}\EnergyFunc(\acurve).
    \]
\end{theorem}
\begin{proof}
    Let $\alpha=\inf_{\acurve\in \LpSpace^2([0,1],\LieAlgebra)}\EnergyFunc(\acurve)$ and let $(\acurve_k)$ be a sequence such that $\EnergyFunc(\acurve_k)\to\alpha$.
    Since $\lambda>0$, the sequence $(\acurve_k)$ is bounded. 
    Furthermore, since $\LieAlgebra$ is finite dimensional, $\LpSpace^2([0,1],\LieAlgebra)$ is separable, and it follows from the Banach--Alaoglu theorem that there is a subsequence such that $(\acurve_k)$ converges weakly to some $\tilde\acurve\in\LpSpace^2([0,1])$.
    By \ref{prp:SolWeakCont}, we have that $\gcurve^{\acurve_k}(1)\to\gcurve^{\tilde\acurve}(1)$.
    Now, the identity $\EnergyFunc(\tilde\acurve)=\alpha$ follows directly from the following calculation:
    \begin{multline*}
        \alpha = \lim_{k \to \infty} \EnergyFunc(\acurve_k)
        \geq\liminf_{k \to \infty} \opStyle{D}(\GroupAction(\gcurve^{\acurve_k}(1),\template) ,\target)+\lambda\liminf_{k \to \infty} \|\acurve_k\|^2_{2}
        \\
        \geq\opStyle{D}(\GroupAction(\gcurve^{\tilde\acurve}(1),\template) ,\target)+\lambda\liminf_{k \to \infty}\|\acurve_k\|^2_{2}
        \geq\opStyle{D}(\GroupAction(\gcurve^{\tilde\acurve}(1),\template) ,\target)+\lambda\|\tilde\acurve\|^2_{2}
        =\EnergyFunc(\tilde\acurve)\geq \alpha.
    \end{multline*}
\end{proof}

 \bibliographystyle{siamplain}
 \bibliography{refs}

\end{document}